%% file: DP_UPAS_arXiv.tex
\documentclass[12pt,draftcls,onecolumn]{IEEEtran}

\input{my_macros.tex}

\usepackage{algpseudocode}
\usepackage[normalem]{ulem}
\newcommand{\p}{\mathrm{p}}
\renewcommand{\d}{\mathrm{d}}
\renewcommand{\st}{\mathrm{s.t.}}

\begin{document}

\title{Downlink Precoding for DP-UPA FDD Massive MIMO via Multi-Dimensional Active Channel Sparsification}

\author{\IEEEauthorblockN{Han Yu, Xinping Yi, and Giuseppe Caire
}
\thanks{H. Yu and X. Yi are with Department of Electrical Engineering and Electronics, University of Liverpool, L69 3BX, United Kingdom. Email: \{han.yu, xinping.yi\}@liverpool.ac.uk.}
\thanks{G. Caire is with Communications and Information Theory Group (CommIT), Technical University of Berlin, 10587, Berlin, Germany. Email: \{caire\}@tu-berlin.de.}
}

\maketitle

\begin{abstract}
In this paper, we consider user selection and downlink precoding for an over-loaded single-cell massive multiple-input multiple-output (MIMO) system in frequency division duplexing (FDD) mode, where the base station is equipped with a dual-polarized uniform planar array (DP-UPA) and serves a large number of single-antenna users. 
Due to the absence of uplink-downlink channel reciprocity and the high-dimensionality of channel matrices, it is extremely challenging to design downlink precoders using closed-loop channel probing and feedback with limited spectrum resource.
To address these issues, a novel methodology -- active channel sparsification (ACS) -- has been proposed recently in the literature for uniform linear array (ULA) to design sparsifying precoders, which boosts spectral efficiency for multi-user downlink transmission with substantially reduced channel feedback overhead. Pushing forward this line of research, we aim to facilitate the potential deployment of ACS in practical FDD massive MIMO systems, by extending it from ULA to DP-UPA with explicit user selection and making the current ACS implementation simplified.
To this end, by leveraging Toeplitz matrix theory, we start with the spectral properties of channel covariance matrices from the lens of their matrix-valued spectral density function. Inspired by these properties, we extend the original ACS using scale-weight bipartite graph representation to the matrix-weight counterpart. Building upon such matrix-weight bipartite graph representation, we propose a multi-dimensional ACS (MD-ACS) method, which is a generalization of original ACS formulation and is more suitable for DP-UPA antenna configurations. The nonlinear integer program formulation of MD-ACS can be classified as a generalized multi-assignment problem (GMAP), for which we propose a simple yet efficient greedy algorithm to solve it. Simulation results demonstrate the performance improvement of the proposed MD-ACS with greedy algorithm over the state-of-the-art methods based on the QuaDRiGa channel models.
\end{abstract}


\section{Introduction}

Massive multiple-input multiple-output (MIMO) has been demonstrated, in both theory and practice, as one of the major performance boosters for the next generation (5G and beyond) wireless communication systems \cite{larsson2014massive,bjornson2017massive}. Operating massive MIMO in time division duplex (TDD) mode is technically favorable, because the inherent uplink-downlink channel reciprocity makes it convenient to reconstruct the downlink channel vectors directly from the uplink pilot observations without requiring downlink training. Nevertheless, from the mobile network operators’ standpoint, the frequency division duplex (FDD) mode seems more preferable, as the current wireless systems are mainly operating in FDD mode, for which a lot of resource has been invested (e.g., in acquiring the spectrum), and FDD systems show a much better performance in scenarios with symmetric data traffic and delay-sensitive applications. 

For FDD massive MIMO, however, the uplink-downlink channel reciprocity does not hold in general, due to the large uplink/downlink frequency separation that exceeds the fading coherence bandwidth. 
The base stations (BSs) have to probe the downlink channels via pilot training and ask for channel information feedback from the users. The high-dimensional channel vectors (due to the large number of antennas) incur prohibitively expensive feedback overhead and therefore result in inevitable performance degradation provided the limited channel coherence time and bandwidth.
Recently, a fast-growing number of techniques have been proposed to make FDD as competitive as TDD systems by reducing downlink training and uplink feedback overhead, e.g., joint spatial division and multiplexing (JSDM) \cite{JSDM, JSDM-UG}. 
It has been observed that the high-dimensional channel vectors admit a sparse representation in the angular/beam domain, such that they could be efficiently represented by low-dimensional ones \cite{JSDM, JSDM-UG,rao2014distributed, gao2015spatially}. As such, the pilot dimension of downlink training and the feedback overhead can be substantially reduced.

By exploiting the sparse representation in the beam domain, a number of techniques have been proposed.
First, the compressed sensing (CS)-inspired methods (e.g., \cite{rao2014distributed, gao2015spatially, ding2018dictionary}) exploit this sparse representation to reconstruct the channel vectors at the BS using the compressed downlink pilot signals fed back from users. In particular, it allows each user to obtain the compressive measurements of the probing signals locally, and feed back to the BS, so that the BS can jointly recover the channel vectors using CS techniques \cite{rao2014distributed}. 
{Although effective to some extent, these CS-based techniques rely highly on the knowledge of channel vector sparsity order (i.e., the number of significant elements), and probably fail to reconstruct downlink channels reliably when the devoted pilot dimension is less than the sparsity order. }
Second, channel reconstruction by exploiting the second-order statistics has attracted more and more attention \cite{xie2018channel, miretti2018fddmassive, haghighatshoar2018multi}. Instead of feeding back compressed measurements of pilot signals, these approaches leverage the angular domain reciprocity to reconstruct downlink channel covariance matrices from the uplink training. 
Third, there is a new trend of using deep learning to predict downlink channel from the observations of uplink training \cite{wen2018deep, wang2018deep, jang2019deep, arnold2019enabling}. The basic idea is to build up a mapping from uplink channel vectors to downlink ones by using an over-parameterized deep neural network. In principle, the deep neural networks with a sufficiently large number of parameters are able to approximate any complicated functions, as long as the training dataset is large enough. Nevertheless, there are still many challenges to design an efficient deep neural network for channel reconstruction.

More recently, channel reconstruction methodologies using second-order channel statistics have been advanced by e.g, \cite{B.K, khalilsarai2020dual, khalilsarai2020structured, liu2020statistical}, which aim to reconstruct downlink channel for FDD massive MIMO by exploiting the  angular scattering function (ASF) reciprocity. This technique relies on the key assumption of the reciprocity of the ASF -- it assumes that the ASF is frequency-invariant over both uplink and downlink frequency bands \cite{haghighatshoar2018multi}. It consists of two major components: (1) Acquiring downlink channel support information in the angular domain from uplink channel training by exploiting uplink/downlink angular domain reciprocity; and (2) exploiting structural properties of such support information to design efficient downlink probing and uplink feedback schemes.
As the channel angular support (i.e., non-zero elements) information is contained in the covariance matrix, its acquisition can be done by estimating downlink covariance matrices from the uplink ones, followed by the angular supports extraction. The channel support information of all users establishes a beam-user association (that can be modeled by a bipartite graph), in which the support of a user’s channel vector indicates the corresponding beams that can be utilized to serve this user. Such a beam-user association can be exploited for intelligent beam-user assignment that leads to artificially sparsified users’ channels. The active channel sparsification (ACS) will finally help reduce the pilot dimension for downlink probing, while allowing for simultaneous multiple-access of a large number of users using spatial multiplexing. 

However, the ACS methodology is still facing some challenges in the potential deployment in the practical massive MIMO systems. On one hand, dual-polarized uniform planar array (DP-UPA) is commonly used in the practical systems, although attempts have been made to extend from ULA to DP-ULA \cite{khalilsarai2020dual}. On the other hand, although the current ACS implementation using the mixed integer linear program (MILP) formulation is elegant in theory, its computational complexity scales as the number of antennas and users. To address these issues, in this paper, we consider to extend the ACS formulation from ULA/DP-ULA to DP-UPA by leveraging a matrix-weight bipartite graph representation for users' channels.  
By relaxing the original MILP formulation, we come up with a new nonlinear integer program (NIP) formulation. Hence, we propose a greedy algorithm to solve the NIP problem in an efficient way. 
Specifically, our contributions are summarized as follows.
\begin{itemize}
\item The channel covariance matrix of massive MIMO with DP-UPA antennas can be recognized as a doubly block Toeplitz matrix. By leveraging Toeplitz matrix theory, we characterize the spectral properties of channel covariance matrices by investigating their matrix-valued spectral density function, which is also referred to as angular scattering function \cite{B.T}. There exhibits some sparsity in the spectral density function when the angular spread is narrow under the context of DP-UPA massive MIMO scenarios.
\item Inspired by these properties, we extend channel representation of ACS using bipartite graph from the original scale-weight to the matrix-weight counterpart. The matrix-weight bipartite graph establishes the association between block beams (correspond to dual-polarized antenna) and users according to the asymptotic block diagonalization of the channel covariance matrices. Building upon the matrix-weight bipartite graph representation, we propose a multi-dimensional ACS (MD-ACS) method, which is a generalized version of original ACS formulation and is more suitable for DP-UPA antenna configurations. The MD-ACS can be formulated as a generalized multi-assignment problem, which includes the original ACS formulation (i.e., assignment problem) as a special case.
\item By taking into account the sum rate maximization and multiuser interference control, we reformulate the MD-ACS approach as a nonlinear integer program, for which we propose a simple yet efficient greedy algorithm to solve it. The extensive simulation results using QuaDRiGa channel models demonstrate the superiority of the proposed MD-ACS with greedy algorithm to the state-of-the-art methods, including the recently advanced ACS method concerning DP-ULA antenna configurations.
\end{itemize}

The rest of this paper is organized as follows. In the next section, we describe the channel and system model of the DP-UPA FDD massive MIMO system with downlink training and precoding. In Section III, we study channel covariance matrices through Toeplitz theory, and characterize the spectral properties of the spectral density functions. The proposed MD-ACS is detailed in Section IV, including the review of the original ACS, the matrix-weight graph representation, and the NIP formulation with a greedy algorithm. The numerical results can be found in Section V, followed by the Conclusion in Section VI.

{\bf Notation:} 
We use $x$, $\xv$, and $\Xm$ to represent scalar, vector, and matrix, respectively. 
For any scalar $x$, we denote $\{x_n\}_{n=1}^N \defeq \{x_1,x_2,\dots, x_N\}$.  
For the integer $N$, we denote $[N]\defeq \{1,2,\dots,N\}$. A matrix $\Xm$ is Hermitian if and only if $\Xm=\Xm^{\H}$, where $\Xm^\H$ is the conjugate transpose of $\Xm$. $\trace{(\Xm)}$ denotes the trace of a matrix $\Xm$. 
$\EE\{\cdot\}$ denotes the expectation. The Kronecker and Hadamard products of two matrices $\Xm$ and $\Ym$ are denoted by $\Xm \otimes \Ym$ and $\Xm \odot \Ym$, respectively.
{$\CN(\alpha,\beta)$ denotes the complex normal distribution, where $\alpha$ and $\beta$ are mean (vector) and variance (matrix), respectively. }
$\IM_M$ is the $M\times M$ identity matrix, and $\Fm_M$ is the discrete Fourier transform (DFT) matrices with $[\Fm_{M}]_{p,q}=\frac{1}{\sqrt{M}}e^{-\jmath\frac{2\pi (p-1)(q-1)}{M}}$ for all $p\in [M],q\in[M]$.

\section{Channel and Signal Model}  
\subsection{DP-UPA Channel Model}
We consider a single-cell massive MIMO system where the base station is equipped with an  $M_x\times M_y\times 2$ dual-polarized uniform planar array (DP-UPA) serving $N_U$ single-polarized single-antenna users. {The DP-UPA consists of in total $M=2 M_x M_y$ antenna elements with $M_x$ ports in each column and $M_y$ ports in each row, and for each port there are two polarized antenna elements.} According to 3GPP {TR-36.873} \cite{3GPP}, which is also referred by e.g., \cite{C.Q} and \cite{L.M}, the channel vector $\hv$ of DP-UPA can be represented as
\begin{align}\label{eq:channelV_channelH}
\hv=\begin{bmatrix}\hv_V\\ \hv_H \end{bmatrix} \in \CC^{M \times 1}
\end{align}
where $\hv_V \in \CC^{\frac{M}{2} \times 1}$ and $\hv_H \in \CC^{\frac{M}{2} \times 1}$ correspond to the channel between the vertical ($V$)/horizontal ($H$) antenna and the user, respectively. For notational simplicity, let $q\in\{V,H\}$. Given the angle intervals of azimuth $\Ac$ and elevation $\Bc$, according to the channel model of 3GPP \cite{3GPP}, the $q$-th sub-channel vector can be written as 
\begin{align}
\hv_q=\int_\Bc\int_\Ac \beta_q(\theta, \phi)\gamma_q\av(\theta,\phi)d\theta d\phi
\end{align}
where $\Ac=[\theta_{\min},\theta_{\max}],\Bc=[\phi_{\min},\phi_{\max}]$ and $|\Ac|=2\delta_\theta$ and $|\Bc|=2\delta_\phi$, in which $\delta_\theta$ and $\delta_\phi$ are the angular spread (AS) of azimuth and elevation, respectively; $\beta_q(\theta,\phi)\sim\CN(0,\beta_q)$ denotes the complex gain that is independent and identically distributed (i.i.d.) across paths; $\gamma_q$ is the polarization factor of the $q$-th sub-channel; and    
$\av(\theta,\phi)$ is the steering vector of DP-UPA antenna that possesses the same structure as that of UPA, and it can be written as \cite{3GPP}\cite{C.Q}\cite{lu2020omnidirectional}
\begin{align} \label{eq:steering-vector}
 &\av(\theta,\phi)=\av_y(\theta,\phi)\otimes \av_x(\theta,\phi)
=
\begin{bmatrix}
1\\e^{\jmath\frac{2\pi d_y}{\lambda_w}\sin(\phi)\sin(\theta)}\\ \vdots \\e^{\jmath\frac{2\pi d_y(M_y-1)}{\lambda_w}\sin(\phi)\sin(\theta)}
\end{bmatrix}\otimes
\begin{bmatrix}
1\\e^{\jmath\frac{2\pi d_x }{\lambda_w}\sin(\phi)\cos(\theta)}\\ \vdots \\e^{\jmath\frac{2\pi d_x (M_x-1)}{\lambda_w}\sin(\phi)\cos(\theta)}
\end{bmatrix}
\end{align}
where $d_x$ and $d_y$ are antenna spacing of column and row array respectively, and $\lambda_w$ is the carrier wavelength.
\subsection{Downlink Training and Precoding}
\label{sec:signal-model}
In this paper, we follow the comprehensive framework proposed in \cite[Figure 4]{khalilsarai2020dual}, which consists of (1) uplink pilot transmission, (2) uplink covariance estimation, (3) uplink-downlink covariance transformation, (4) downlink pilot transmission, (5) feeding back pilot measurements, (6) downlink channel estimation, and (7) downlink beamforming.
As our focus in this paper is on the downlink precoding/beamforming, we assume the availability of downlink covariance matrix at the base station via the above steps (1)-(3). In what follows, we briefly reiterate the procedure of (4)-(7) to maintain certain level of self-containedness.

\subsubsection{Downlink Pilot Transmission}
As in \cite{khalilsarai2020dual}, the base station sends a space-time pilot matrix $\Sm\in\CC^{T\times M'}$ to all users through a sparsifying precoder $\Vm_h \in \CC^{M \times M'}$, where $T$ is the number of time slots used for pilot trainsmission, {$M' \le M$ is the dimension after the active sparsification, and the columns of $\Vm_h$ are chosen from an orthogonal matrix that will be specified later.}
As such, the received pilot signal $\yv_i^{\p}$ of $i$-th user  can be written as 
\begin{align}\label{eq:estimated channel}
&\yv^{\p}_i = \Sm \Vm_h^\H \hv_i +\nv,
\end{align}
where $\hv_i\in\CC^{M\times 1}$ is the downlink channel vector of $i$-th user, and $\nv \sim \CN(\zerov, \sigma^2 \IM_M)$ is the additive white Gaussian noise (AWGN). The pilot matrix $\Sm$ is up to design, subject to a total power constraint $\trace(\Sm\Vm_h^\H\Vm_h\Sm^\H)\leq \rho^{\p}T$, where {$\rho^{\p}$ is the pilot signal power in each time slot.} 

\subsubsection{Feeding Back Pilot Measurements} For simplicity, we assume the users feed back their pilot signals $\yv^{\p}_i \in \CC^{T \times 1}$ to the base station in an analog form. The digital feedback with quantization can be implemented according to well-developed techniques (see \cite{love2003equal} and references therein). Due to possible user selection, only the selected users are required to send the pilot signals back to the base station. In doing so, the base station could successfully acquires the perfect pilot signals $\{\yv^{\p}_i \}_{i \in \Sc}$ with {$\Sc$ being the subset of selected users, which will be specified later.}

\subsubsection{Downlink Channel Estimation}
Given the $T \times 1$ pilot signal $\yv^{\p}_i$, we aim to recover the $M \times 1$ channel vector $\hv_i$ with $M > T$, relying on the sparsity of $\hv_i$ in the angular domain.
Following the footstep in \cite{khalilsarai2020dual}, we obtain the estimated channel vector via MMSE estmators as
\begin{align}\label{est-channel}
\hat{\hv}_i = \Rm_{h,i}\Rm_{y,i}^{-1}\yv_i^{\p},
\end{align}
where $\Rm_{h,i} = \EE\{\hv_i(\yv_i^\p)^\H\}= \Rm_i \Vm_h\Sm^\H$, $\Rm_{y,i} = \EE\{\yv_i^\p(\yv_i^\p)^\H\}=\Sm\Vm_h^\H\Rm_i\Vm_h\Sm^\H + \sigma^2 \IM$ with $\Rm_i \defeq \EE [\hv_i \hv_i^{\H}]$ being the downlink channel covariance matrix of user-$i$. 


\subsubsection{Downlink Precoding} With channel estimates, the base station transmit users' data $\{d_i\}_{i \in \Sc}$ through sparsifying precoders $\pv_i \in \CC^{M \times 1}$ for each selected user $i \in \Sc$.
Thus, the received signal under FDD DP-UPA downlink data phase $y_i^\d$ of $i$-th user can be written as 
\begin{align}
y_i^{\d}=\hv_i^\H\pv_id_i+\sum_{j \in \Sc \backslash i}\hv_i^\H\pv_jd_j+n_i
\end{align}
where $n_i\sim\CN(0,\sigma_i^2)$ is the AWGN, and the sparsifying precoder $\pv_i$ will be specified later. As the downlink covariance matrix estimation has been extensively investigated in the literature (e.g, \cite{L.M,khalilsarai2020dual,miretti2018fddmassive}), we place our focus instead on designing the downlink precoder assuming that the downlink channel covariance matrix $\{\Rm_i\}_{i=1}^{N_U}$ is perfectly known at the base station.

\section{Spectral Properties of Covariance Matrix}\label{sec-3}
While some existing works have mentioned the Toeplitz structure of covariance matrices for ULA/UPA massive MIMO (e.g., \cite{JSDM, Yu:ACS}), the extension to DP-UPA has not been fully understood. In what follows, we will inspect the structural properties of downlink channel covariance matrices $\{\Rm_i\}_{i=1}^{N_U}$ for DP-UPA massive MIMO through the lens of Toeplitz matrix theory.

\subsection{Toeplitz Matrix Theory}

Before proceeding further, we first introduce the definitions related to Toeplitz matrix \cite{gray2006toeplitz} and its extension to block Toeplitz matrix \cite{B.T} and doubly Toeplitz matrix \cite{pa1996atheorem,oudin2008asymptotic}.

Given a sequence of scalars 
$\{t_{-n+1}, \dots, t_{-1}, t_0, t_1 \dots, t_{n-1}\}$, an $n\times n$ matrix $\Tm_n$ is a Toeplitz matrix if $[\Tm_n]_{i,j}=t_{i-j}$ for all $i,j\in[n]$. Similarly, given a sequence of $M_1 \times M_2$ matrices 
$\{\Tm_{-n+1}, \dots, \Tm_{-1}, \Tm_0, \Tm_1 \dots, \Tm_{n-1}\}$
, an $nM_1\times nM_2$ matrix $\Bm_n$ is a block Toeplitz matrix if the $(i,j)$-th $M_1 \times M_2$ submatrix $[\Bm_n]_{i,j}=\Tm_{i-j}$ for all $i,j\in[n]$. In particular, if $\Tm_{m}$ is an $N \times N$ Toeplitz matrix for all $-n+1 \le m \le n-1$, then $\Bm_n$ is an $nN \times nN$ doubly Toeplitz matrix, also known as Toeplitz-block-Toeplitz (TBT) matrix (i.e., block Toeplitz matrix with Toeplitz blocks). Further, if $\Tm_{m}$ is an $NM_1 \times NM_2$ block Toeplitz matrix for all $-n+1 \le m \le n-1$, then $\Bm_n$ becomes an $nNM_1 \times nNM_2$ doubly block Toeplitz matrix. It can be viewed as a Toeplitz matrix with each element being block Toeplitz matrices, or a doubly Toeplitz matrix with each element being a general matrix. Throughout this paper, we consider Hermitian matrix, that is $t_{-i}=t_{i}^*$ for Toeplitz matrix and $\Tm_{-i} = \Tm_{i}^\H$ for block Toeplitz matrix.

The circulant, block circulant, doubly circulant, doubly block circulant matrices can be similarly defined as their Teoplitz counterparts, where the only difference is the circular operation using {the modulo operator}$\!\mod\!\!$, i.e., for the scalar sequence $[\Cm_n]_{i,j}=c_{(i-j) \!\mod n}$ for all $i,j\in[n]$ and for the matrix sequence $[\Bm_n]_{i,j}=\Cm_{(i-j) \!\mod n}$ for all $i,j\in[n]$. Apparently, the (doubly block) circulant matrix is the special case of (doubly block) Toeplitz matrix. Given an $n \times n$ circulant matrix $\Cm_n$, it can be diagonalized by DFT matrix, i.e., $\Cm_n = \Fm_n \Lambdam \Fm_n^\H$ with $\Lambdam$ being a diagonal matrix. For an $nN \times nN$ doubly circulant matrix, it can be diagonalized by 2D-DFT matrix $\Fm_n \otimes \Fm_N$. For an $nM_1 \times nM_2$ block circulant matrix $\Bm_n$, it can be block-diagonalized by $\Bm_n = (\Fm_n \otimes \IM_{M_1}) \Sigmam (\Fm_n \otimes \IM_{M_2})^\H$ where $\Sigmam$ is an $nM_1 \times nM_2$ block diagonal matrix, with each block being an $M_1 \times M_2$ matrix.

When $n$ tends to infinity, each Toeplitz matrix can be associated with a generating function, which is a continuous and periodic function \cite{gray2006toeplitz,B.T,pa1996atheorem,oudin2008asymptotic}. 
For instance, the Hermitian Toeplitz matrix $\Tm_n$ can be generated by a real function $F: [-1/2, 1/2] \mapsto \RR$, i.e.,
$F(\omega)=\sum_{k=-\infty}^{\infty}t_ke^{\jmath 2\pi k \omega}$. Similarly, the Hermitian block Toeplitz matrix $\Bm_n$ can be generated by a matrix-valued real function $F: [-1/2,1/2] \mapsto \RR^{M_1 \times M_2}$, i.e., $F(\omega)=\sum_{k_1=-\infty}^{\infty}\Tm_{k_1}e^{\jmath 2 \pi k_1\omega}$.
Further, for the Hermitian doubly block Toeplitz matrix, it can be generated by  $F: [-1/2,1/2]^2 \mapsto \RR^{M_1 \times M_2}$, i.e., 
\begin{align}\label{eq:TBT_G}
&F(\omega_1,\omega_2)=\sum_{k_1=-\infty}^{\infty}\sum_{k_2=-\infty}^{\infty}\Tm_{k_1,k_2}e^{\jmath 2 \pi k_1\omega_1}e^{\jmath 2 \pi k_2\omega_2}.
\end{align}
The circulant conterparts share the same generating functions.

\subsection{Spectral Properties}
By leveraging the Toeplitz matrix theory, we inspect the spectral properties of channel covariance matrix through a function analysis perspective. In particular, instead of looking into the channel covariance matrix, we investigate its spectral density in the angular domain. This is underpinned by the following lemma.
\begin{lemma}
\label{lemma:structure_cov}
The channel covariance matrix  $\Rm$ of DP-UPA massive MIMO can be represented, subject to row/column permutation, as a Hermitian doubly block Toeplitz matrix $\hat{\Rm}$, which can be asymptotically block diagonalized by an orthongal matrix
\begin{align}
\Vm = \Fm_{M_x} \otimes \Fm_{M_y} \otimes \IM_2,
\end{align}
as $M_x,M_y\to \infty$, where $\Fm_{n}$ is an $n \times n$ DFT matrix, and the block-diagonal submatrices are uniformly sampled from the matrix-valued spectral density function, i.e.,
\begin{align} 
&\Sigmam(\omega_1,\omega_2)=\sum_{m_1=-M_y+1}^{M_y-1}\sum_{m_2=-M_x+1}^{M_x-1}[\hat{\Rm}]_{m_1,m_2}e^{ \jmath 2 \pi (m_1\omega_1+m_2\omega_2)}
\end{align}
with $[\hat{\Rm}]_{m_1,m_2}$ being a $2 \times 2$ submatrix of $\hat{\Rm}$.
\end{lemma}
\begin{proof}
See Appendix \ref{proof:lemma-toeplitz}.
\end{proof}

\begin{remark}
The $2 \times 2$ matrix-valued spectral density function $\Sigmam(\omega_1,\omega_2)$ over $(\omega_1,\omega_2) \in [-1/2,1/2]^2$ is the generating function of the doubly block Toeplitz matrix $\hat{\Rm}$. As row/column permutation does not change spectral properties, $\Sigmam(\omega_1,\omega_2)$ is the spectral density function of channel covariance matrix $\Rm$ over the two-dimensional angular domain $ [-1/2,1/2]^2$. Similar to the ULA massive MIMO, we can also transform the signals from spatial to angular domain to exploit possible (block) sparsity of the spectral density.
The columns of the DFT-type orthogonal matrices have been widely used as the common basis for Toeplitz, block Toeplitz and TBT matrix in massive MIMO such as precoder design for ULA\cite{X.R,khalilsarai2020dual} and UPA\cite{Yu:ACS} array, and pilot decontamination\cite{Haifan,you2015pilot,Z.C}. 
\end{remark}

Equipped with Lemma \ref{lemma:structure_cov}, we are able to inspect the spectra of channel covariance matrix $\Rm$ through its spectral density function $\Sigmam(\omega_1,\omega_2)$. As such, we come up with the sparsity properties of DP-UPA antennas in the angular domain, as in Theorem \ref{theorem:sparsity}.

\begin{theorem}
\label{theorem:sparsity}
The spectral density function $\Sigmam(\omega_1,\omega_2)$ has a compact support over the two-dimensional frequencies  $(\omega_1,\omega_2) \in [-1/2,1/2]^2$, i.e.,
\begin{align}
\Sigmam(\omega_1,\omega_2) = \mathbf{0}, \quad \mathrm{ if~} (\omega_1, \omega_2) \notin \left[-\frac{d}{\lambda_w}z_1^{\max},\frac{d}{\lambda_w}z_1^{\min}\right] \times \left[-\frac{d}{\lambda_w}z_2^{\max},\frac{d}{\lambda_w}z_2^{\min}\right]
\end{align}
where $z_i^{\min}$ and $z_i^{\max}$ depend on  a fixed AOA $\theta_c,\phi_c$ and AS $\Delta_1, \Delta_2$.
\end{theorem}
\begin{proof}
See Appendix \ref{proof:theorem-sparsity}.
\end{proof}
\begin{remark}
Theorem \ref{theorem:sparsity} is a generalization of the compact properties of spectral density function from ULA reported in \cite{JSDM} to DP-UPA antenna configurations.
In contrast with the ULA, the spectral density function of DP-UPA is $2\times 2$ matrix-valued because of the dual-polarization. In addition,
for UPA and DP-UPA antennas, the compact supports of $\Sigmam(\omega_1,\omega_2)$ could be more dispersed, thanks to the two-dimensional array. This enables UPA-type antennas to server more users without causing severe pilot contamination or multiuser interference.

Thanks to the high resolution of large-scale antenna arrays, the azimuth and elevation AoAs are usually limited within a narrow range \cite{jaeckel2014quadriga}, so that $z_i^{\max}$ and $z_i^{\min}$ are confined within small intervals in $[-1,1]$. As such, the compact support only covers a limited range of frequency range, and thus the spectral density exhibits sparsity properties in the angular domain. 
To illustrate the above points, we plot the spectral density of covariance matrices of ULA, UPA, and DP-UPA with the same number of antennas, using channels generated by QuaDRiGa \cite{jaeckel2014quadriga} (See Section \ref{sec:numerical} for the configurations). In particular,
Figure \ref{Com_set_1} shows the normalized diagonal elements of DFT-diagonalized covariance matrices for $128\times 1$ ULA, $16\times 8$ UPA, and $8\times 8\times 2$ DP-UPA, respectively. 
It can be observed that ULA has one single yet wide support, and UPA and DP-UPA have multiple narrow supports. Additionally, for DP-UPA, it exhibits the block support where the supports appear in pair, which agrees with the $2 \times 2$ matrix-value spectral density function.
\end{remark}

\begin{figure}[t]
\centering
\includegraphics[width= 3.5in,angle=0]{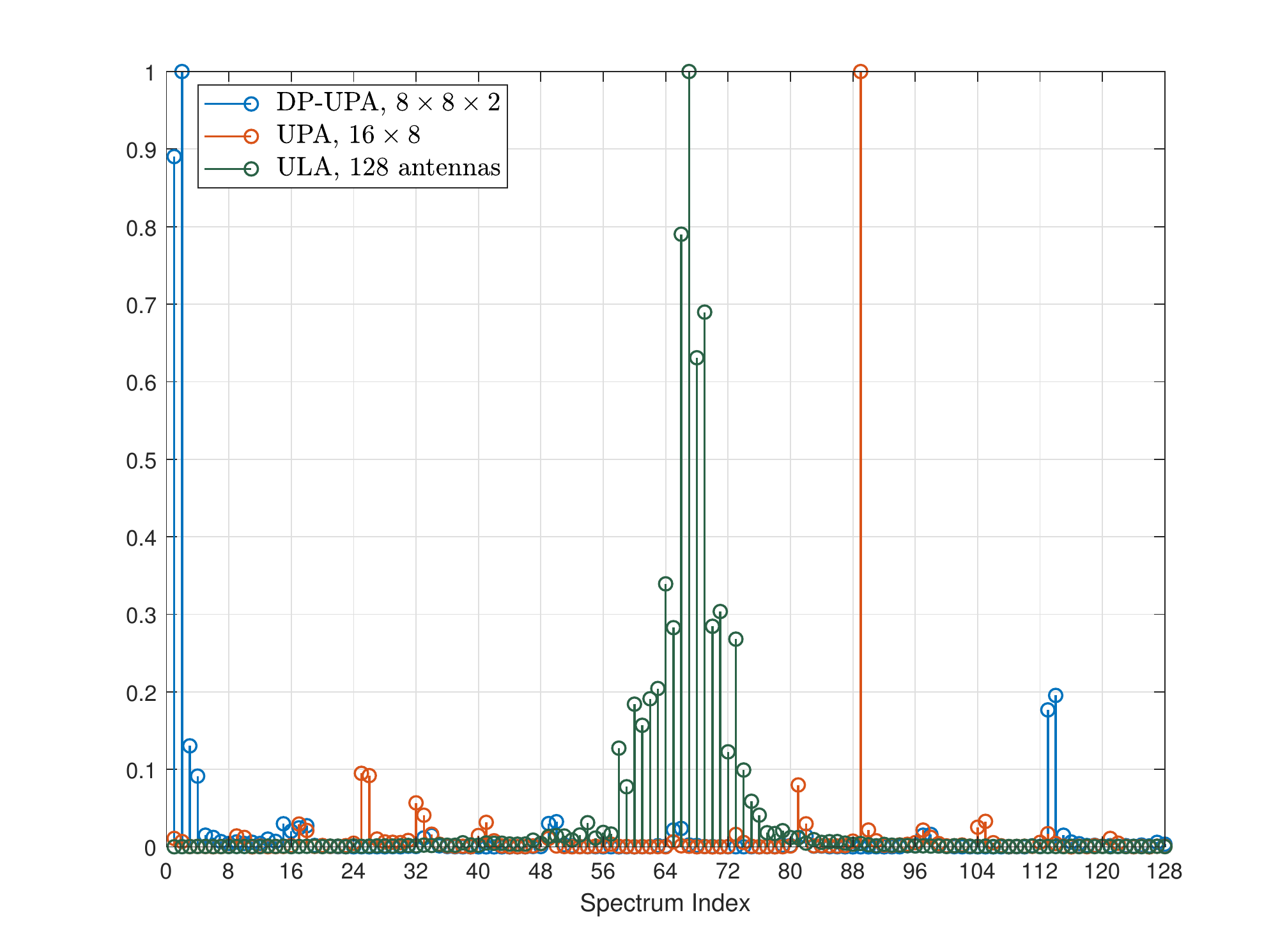}
\vspace{-10pt}
\caption{The normalized spectra of covariance matrices for $128 $ ULA, $16\times 8$ UPA, and $8\times 8\times2$ DP-UPA antennas.}
\label{Com_set_1}
\vspace{-20pt}
\end{figure}

\section{Multi-Dimensional Active Channel Sparsification}
\label{sectionIV}

Inspired by the matrix-valued spectral density function, we extend the active channel sparsification to multi-dimensional scenarios, with a generalized optimization problem formulation. In what follow, we first summarize the merit of active channel sparsification proposed in \cite{B.K}, followed by a matrix-weight graph representation, and finally we bridge the general optimization problem formulation to an existing problem in combinatorial optimization.

\subsection{Preliminary: Active Channel Sparsification}
For the sake of self-containedness, we briefly introduce the main idea of active channel sparsification proposed in \cite{B.K}, with the focus on ULA antenna geometry.
\subsubsection{Channel Representation} Each user's channel is represented as a weighted sum of a set of vectors chosen from common bases. These common basis vectors are also referred as to virtual beams in the angular/beam domain. Usually, for ULA antenna setting, the columns of discrete Fourier transform (DFT) matrix are adopted as common basis vectors. This is underpinned by the fact that the channel covariance matrix of ULA antenna is a Toeplitz matrix, which asymptotically approximates the circulant matrix that can be diagonalized by DFT matrix.

Such a representation ensures that all users are represented in the same vector space, such that beam selection can be alternatively done by switching on/off the basis vectors. If we use $x_m$ to denote the status of the beam-$m$, we have
\begin{align}
x_m = \left\{ \Pmatrix{1, \quad \text{if beam-$m$ is selected,}\\ 0, \quad \text{otherwise.}} \right.
\end{align}
In a similar way, we can also impose a binary variable $y_i$ to denote user selection, i.e.,
\begin{align}
y_i = \left\{ \Pmatrix{1, \quad \text{if user-$i$ is selected,}\\ 0, \quad \text{otherwise.}} \right.
\end{align}
Hence, after adopting the user and beam selection strategy, the estimate of $i$-th user's channel $\hv_i$, as in \eqref{est-channel} in Section \ref{sec:signal-model}, can be asymptotically written as 
\begin{align}
\hat{\hv}_i \approx \sum_{m=1}^M x_m {\iota_{i,m}} \sqrt{w_{i,m}}\fv_m
\end{align}
if user-$i$ is selected, i.e., $ y_i=1$, where $\fv_m$ is the $m$-th common basis vectors used for channel representation, $w_{i,m}$ is the corresponding coefficient that can be estimated by downlink training, {and $\iota_{i,m}$ is a random variable}. Usually, in the ULA setting, $\fv_m$ comes from the columns of the DFT matrix $\Fm_M$ (e.g., \cite{B.K, T.M}), and in the $M_x \times M_y \times 2$ DP-UPA setting, $\fv_m$ is usually from the columns of the common basis $\IM_2\otimes\Fm_{M_x}\otimes \Fm_{M_y}$ (e.g., \cite{khalilsarai2020dual}). Accordingly, the sparsifying precoder $\Vm_h$ in \eqref{eq:estimated channel} can be specified as the collection of basis vectors $\{\fv_m: x_m=1\}$. 
\subsubsection{Graph Representation} 
To describe the interaction between beam and user selection, we can construct a weighted bipartite graph, where beams are on one side and users are on the other side, and a beam and a user is connected if the beam contributes to channel representation of such user. By such bipartite graph representation, we establish the user-beam association with respect to the weighted combinations of channel representation. 

For the readers' reference, we introduce some graph definitions. Consider an undirected bipartite graph $\Gc=(\Uc,\Vc,\Ec)$ with two vertex sets $\Uc$ and $\Vc$, and an edge set $\Ec$. For any $u \in \Uc$ and $v \in \Vc$, $e=(u,v) \in \Ec$ if and only if $u$ and $v$ are connected with an edge $e$.
The neighborhood of a vertex $v$ is the set of nodes $u \in \Uc$ such that $(u,v) \in \Ec$, i.e., $\Nc_{\Gc}(v) \defeq \{u \in \Uc: (u,v) \in \Ec \}$.
The degree of a vertex $v$ is the number of nodes in the neighborhood of $v$, i.e., $\text{deg}_{\Gc}(v) \defeq \abs{\Nc_{\Gc}(v)}$ where $\abs{\Nc}$ is the cardinality of the set $\Nc$.
The adjacency matrix $\Am$ of the bipartite graph $\Gc=(\Uc,\Vc,\Ec)$ is a binary matrix, where $\Am_{i,j}=1$ if $(i,j) \in \Ec$ and 0 otherwise.
The {Bipartite matching} of the bipartite graph $\Gc=(\Uc,\Vc,\Ec)$ is a subset of edges $\Mc_{\Gc} \subset \Ec$ such that there are not edges in $\Mc_{\Gc}$ sharing the same vertex.

\subsubsection{Beam/user Selection} The aim of beam/user selection is to switch on/off beams and users to avoid beam overlapping among selected users, in order to achieve the maximum multiplexing gain (i.e., prelog of the sum rate expression). The optimization problem was given in \cite{B.K} as
\begin{subequations}\label{*}
\begin{align}
(\Pc_1): \quad
 \max & \quad \left| \Mc_{\Gc'}\right|\\
\label{eq:acs-degree_constraint}
\mathrm{s.t.} & \quad  \text{deg}_{\Gc'}(u_i) \leq T, \quad \forall u_i\in\Uc' \\
\label{eq:acs-power_constraint}
& \quad \sum_{b_m \in \Nc_{\Gc'}(u_i)} w_{i,m} \ge P, \quad \forall u_i\in\Uc'
\end{align}
\end{subequations}
where $\left| \Mc_{\Gc}\right|$ is the maximum cardinality bipartite matching number of the selected subgraph $\Gc'=(\Bc',\Uc',\Ec')$, and the degree constraint \eqref{eq:acs-degree_constraint} guarantees that for each the selected user $u_i \in \Uc'$ the number of beams to represent this user's channel vector is no more than $T$, and the power constraint \eqref{eq:acs-power_constraint} is to ensure that for each selected user $u_i \in \Uc'$ the sum power of representing beams is no less than $P$. The degree and power constraints ensure that each selected user should have a sufficient number of (but not too many) representing beams selected, so that those beams with little contribution to a user's channel representation can be switched off.

As usually there are much more beams than users, by intuition, the maximum cardinality bipartite matching tends to select all users and only the users have severe conflicting representing beams will be unselected. As such, there is only implicit user selection through beam selection. 

\subsubsection{Casting as an MILP}
By establishing the equivalence between the multiplexing gain of the users' effective channel and the maximum cardinality bipartite matching of the graph representation, the objective of ACS can be solved by finding the solutions to an MILP \cite{B.K} involving two sets of binary variables $\{x_m\}_{m=1}^{M}$ and $\{y_i\}_{i=1}^{N_U}$, and a set of continuous ones, $\{z_{i,m}\}_{i=1,m=1}^{N_U,M}$ i.e., 
\begin{subequations}\label{eq:ACS-MILP}
\begin{align}
(\Pc_1'): \max_{x_m,y_i, z_{i,m}} & \sum_{b_m\in\Bc}\sum_{u_i\in\Uc} z_{i,m} \label{eq:MILP-0}\\ 
\mathrm{s.t.} \quad 
& z_{i,m} \le [\Am]_{i,m}, \quad \forall b_m\in \Bc,\ u_i \in \Uc  \label{eq:MILP-1}\\
& \sum_{u_i \in \Uc} z_{i,m} \leq x_m, \quad \forall b_m\in \Bc \label{eq:MILP-2}\\
& \sum_{b_m \in \Bc} z_{i,m}\leq y_i, \quad u_i\in\Uc\label{eq:MILP-3}\\
& \sum_{b_m \in \Bc} [\Am]_{i,m}z_{i,m} \leq T y_i + M(1-y_i), \quad \forall u_i\in \Uc \label{eq:MILP-4}\\
& P y_i \le \sum_{b_m \in \Bc} [\Wm]_{i,m} x_{m}, \quad \forall  u_i \in \Uc  \label{eq:MILP-5}\\
& x_m \le \sum_{u_i \in \Uc} [\Am]_{i,m} y_i, \quad \forall b_m \in \Bc  \label{eq:MILP-6}\\
& x_m,y_i\in\{0,1\}\quad \forall u_i\in\Uc,\ b_m\in \Bc \label{eq:MILP-7}\\
& z_{i,m}\in[0,1]\quad \forall u_i\in \Uc, \ b_m\in\Bc
\end{align}
\end{subequations}
where binary matrix $\Am$ is the adjacency matrix of graph $\Gc$, and $[\Wm]_{i,m}=w_{i,m}$ indicates the contribution of the block beam $m$ to the $i$-th user's channel representation.  By such an MILP formulation, we can adopt off-the-shelf solvers to find a {feasible} solution $\{x^*_m\}_{m=1}^{M}$, $\{y^*_i\}_{i=1}^{N_U}$ and $\{z^*_{i,m}\}_{i=1,m=1}^{N_U,M}$ efficiently, where the selected beams and users are indicated by $\{m: x^*_m=1\}$ and $\{i: y^*_i=1\}$ respectively in the optimal solution yielded by the MILP.

\subsection{Matrix-weight Bipartite Graph Representation}\label{sec-GAP}

From Section \ref{sec-3}, the covariance matrix $\hat{\Rm}_i$ can be asymptotically block-diagonalized by
\begin{align} \label{eq:blk-diag}
\lim_{M_x,M_y \to \infty}\hat{\Rm}_i &= (\Fm_{M_y}\otimes \Fm_{M_x}\otimes \IM_2) \Sigmam_i (\Fm_{M_y}\otimes \Fm_{M_x}\otimes \IM_2)^\H\\
&= \sum_{m_1=1}^{M_y} \sum_{m_2=1}^{M_x} (\fv_{v,m_1} \otimes \fv_{h,m_2} \otimes \IM_2) \Sigmam_i(m_1,m_2)  (\fv_{v,m_1} \otimes \fv_{h,m_2} \otimes \IM_2)^\H
\end{align}
where $\fv_{v,m}$ and $\fv_{h,m}$ are the $m$-th column of DFT matrices $\Fm_{M_y}$ and $\Fm_{M_x}$, respectively, and $\Sigmam_i(m_1,m_2)$ is the $(M_y(m_1-1)+m_2)$-th diagonal block matrix of $\Sigmam_i$.

Instead of using a vector to represent a virtual beam in the ULA and UPA settings, here we use a $M \times 2$ submatrix $\Vm_{m_1,m_2} \defeq \fv_{v,m_1} \otimes \fv_{h,m_2} \otimes \IM_2$ to represent a virtual \emph{cross-polarized} block beam.
Similarly,
we can represent all users' channels by a bipartite graph with matrix-valued weights, where the cross-polarized block beams $\{\Vm_{m_1,m_2}, m_1 \in [M_y], m_2 \in [M_x]\}$ on one side and the users on the other side, and the beams and users are connected with edges of matrix-valued weights $\{\Sigmam_i(m_1,m_2), m_1 \in [M_y], m_2 \in [M_x]\}$. For notational simplicity, we use $[\Sigmam_i]_m$ to denote the matrix-valued weight for $m \in [M/2]$ corresponding to some $(m_1,m_2)$.

\begin{figure}[t]
\centering
\includegraphics[width= 6in,angle=0]{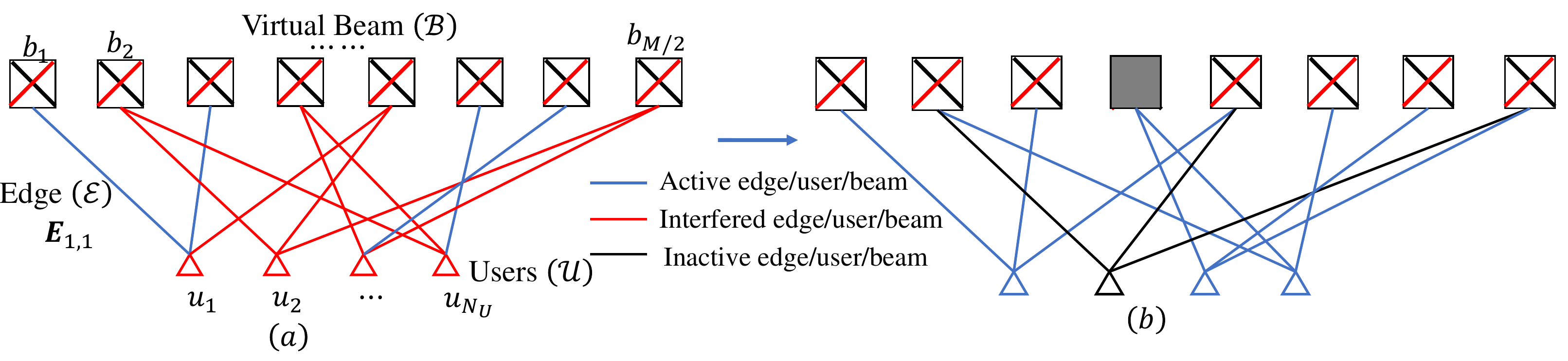}
\caption{Matrix-weight bipartite graph for channel representations, where the virtual block beams are denoted by a square with crossed lines (cf. cross-polarized antenna elements), the users are denoted by triangles, and the weights between beams and users $\Em_{i,m}$ are $2 \times 2$ matrices. (a) Channel representations from different users are overlapping in the sense that they share some common block beams (indicated by red edges) to represent their respective channels. (b) After active channel sparsification applied, some block beams (marked in gray) and users (marked in black) are switched off to avoid channel overlapping, so that the remaining users are not overlapping on active block beams.}
\label{fig_2}
\vspace{-15pt}
\end{figure}

We refer to the scalar-weight graph representation of previous ACS formulation as single-dimension, and the matrix-weight one as multi-dimension bipartite graph representation.
In particular,
we represent the users' channel covariance matrices in respect of the block beams in a matrix weighted bipartite graph $\Gc$ as in Fig. \ref{fig_2}, where a block beam corresponds to a pair of cross-polarized antennas. For notational simplicity, we index the block beams as $m \in [M/2]$. We define the matrix weighted bipartite graph $\Gc=(\Bc,\Uc,\Ec)$, in which the block beams $b\in\Bc$ is on one side and users $u\in\Uc$ on the other side. 
 Therefore, a beam $b_m$ and a user $u_i$ are connected with an edge $(b_m,u_i)\in\Ec$ if $\Am_{i,m}=1$. It is worth noting that, the weight of edges $(b_m,u_i) \in \Ec$, i.e., $\Em_{i,m}=[\Sigmam_i]_m$ with $m$ corresponding to some $(m_1,m_2)$, is a $2 \times 2$ matrix rather than a scalar. With the block beam and user selection parameters $x_m$ and $y_i$, the estimated channel of $i$-th user, as in \eqref{est-channel} in Section \ref{sec:signal-model}, can be approximately written as  
\begin{align}
\hat{\hv}_i \approx \sum_{m=1}^{M/2}x_m \Vm_m \left([\Sigmam_i]_m\right)^{\frac{1}{2}} \iotav_{i,m} 
 \end{align}
where $\Vm_m \in \CC^{M \times 2}$ corresponds to the block basis vectors $\Vm_{m_1,m_2}$ with $m$ corresponding to some $(m_1,m_2)$, 
and $\iotav_{i,m} \in \CC^{2 \times 1}$ is a random vector. Similarly, the sparsifying precoder $\Vm_h$ in \eqref{eq:estimated channel} can be specified as the collection of $\{\Vm_m: x_m=1\}$.

Let us explain the physical meaning of the matrix weighted bipartite graph. Each block beam is illustrated as a pair of crossed lines, in which the cross with red and black lines corresponds to the vertical and horizontal polarization antennas in the DP-UPA array, respectively. As a matter of fact, such a correspondence is resulted from the block-diagonalization of the channel covariance matrix, where it combines the cross-polarized antennas at the same position. The diagonal elements of $\Em_{i,m}\in\CC^{2\times 2}$ represent the channel characteristics of the corresponding antenna, while the off-diagonal elements indicate the channel correlation between the vertical and horizontal antennas due to their cross-polarization. 
In Fig. \ref{fig_2}(a), the edges in red indicate the inter-user spectral correlation between users' channels in the angular domain, which results in potential inter-user interference for multi-user transmission. 
Fig. \ref{fig_2}(b) presents a simple beam and user selection to reduce the possible beam overlapping for activated users.
When actively switching off some block beams and users, the partial channel correlation of the remaining users is reduced, for which there is not any overlap on the activated beams anymore. Note here that the original channels covariance matrices are partially represented by the active block beams only. 

While this may result in partial channel estimation and exploitation, it is expected not to degrade the overall multi-user performance as long as a proper beam and user selection strategy is designed. For the scalar case, it has been evidenced in \cite{B.K} that the beam and user selection by ACS could improve overall performance over the ones without channel sparsification. In what follows, before proceeding with the matrix-weight bipartite graph representation, we take a step back to propose a more general formulation of ACS from the lens of combinatorial optimization.

\subsection{Generalized Multi-dimensional Active Channel Sparsification (MD-ACS)}
Given the above matrix-weight graph representation, we reformulate the original ACS \cite{B.K} in a more general way. The generalization lies in two aspects: one is to extend one-to-one matching to many-to-many matching, the other one is to {generalize scale-weight  (i.e., single-dimensional) to matrix-weight (i.e., multi-dimensional) matching with rate maximization and interference mitigation embedded instead of maximizing multiplexing gain.}

\subsubsection{From One-to-one to Many-to-Many Matching}
In the original formulation in \eqref{*} of ACS, a subgraph $\Gc'$ is selected with active beams $\Bc'$ and users $\Uc'$, and the maximal bipartite matching is constructed in the induced subgraph $\Gc'=(\Bc',\Uc',\Ec')$. It has been shown in \cite{B.K} that the cardinality of the maximal bipartite matching is equal to the multiplexing gain of multi-user transmission. If we take a step back, instead of working on the maximal (one-to-one) bipartite matching in the selected graph, we could consider the many-to-many matching on the original bipartite graph $\Gc$, where a number of beams can be associated to one user, and likewise each beam can serve multiple users as long as inter-user interference is properly controlled. As such, a more general formulation can be given by
\begin{subequations}\label{GACS}
\begin{align}
(\Pc_2): \quad
 \max & \quad w(\Mc^*_{\Gc})\\
\label{eq:degree_constraint}
\mathrm{s.t.} & \quad  \text{deg}_{\Gc}(u_i) \leq \kappa_{b,i}, \quad \forall u_i\in\Uc \\
\label{eq:power_constraint}
& \quad \text{deg}_{\Gc}(b_m) \le \kappa_{u,m}, \quad \forall b_m\in\Bc
\end{align}
\end{subequations}
where $\Mc^*_{\Gc}$ is the set of many-to-many matching, which is a generalization of one-to-one matching,  {$\kappa_{b,i} \le T/2$ is the maximum beams can be assigned to user $i$ to guarantee that channel estimation is feasible \cite{B.K}}, and $\kappa_{u,m}$ the maximum users that can reuse the same beam $m$ so that not much interference is caused one another. In contrast to the one-to-one matching, many-to-may matching allows each vertex on one side to be matched with multiple vertices on the other side.

The above many-to-many weighted matching is equivalent to the generalized multi-assignment problem (GMAP) \cite{park1998lagrangian}, which is a generalized version of the assignment problem corresponding to one-to-one matching.
The GMAP considers to assign a set of tasks to a set of agents. When a task is assigned to an agent, it produces profit and incurs cost. The aim of GMAP is to assign each task to multiple agents, where one agent can conduct multiple tasks, so that the total cost of all tasks is minimized and/or the total profit is maximized.  
Under the context of the multiple beam-user assignment, the above generalized ACS formulation can be reformulated as a GMAP with an integer programming as follows
\begin{subequations}\label{eq:GAP_1}
\begin{align}
( \Pc_2'): \max_{z_{i,m}} \quad & \sum_{m=1}^{M/2} \sum_{i=1}^{N_U} w_{i,m} z_{i,m}  \label{eq:gap-obj}\\
\st \quad &\sum_{i=1}^{N_U}{z_{i,m}} \leq \kappa_u,\  \forall b_m\in\Bc \label{eq:gap-beam} \\
&\sum_{m=1}^{M/2} {z_{i,m}} \leq \kappa_b,\ \forall u_i\in\Uc \label{eq:gap-user}\\
&z_{i,m}\in \{0,1\}, \ \forall b_m\in\Bc, \forall u_i\in\Uc
\end{align}
\end{subequations}
where $z_{i,m}$ is a binary decision variable such that $z_{i,m}=1$ indicates the $m$-th block beam is assigned to $i$-th user, and 0 otherwise; $w_{i,m}$ is the corresponding profit for such an assignment, and it is a function of the matrix-weight $\Em_{i,m}$; $\kappa_u$ and $\kappa_b$ are the maximum number of users and block beams to match each beam and user, respectively. For simplicity, we assume each user (resp. beam) is associated to the same number of beams (resp. users).

\subsubsection{From Single-dimensional to Multi-dimensional Matching}
The generalized formulation in \eqref{eq:GAP_1} reduces the size of the integer program compared to \eqref{*} to a great extent, thanks to the GMAP formulation and the matrix-weight bipartite representation of beam-user association. However, the merits in the original formulation, e.g., multiplexing gain maximization in \eqref{eq:MILP-0} and interference control in \eqref{eq:MILP-5}, are totally lost.

To remedy the above reformulation, in what follows, we integrate the consideration of sum rate maximization into the objective function, especially into the parameters $\{w_{i,m}\}$, and relegate the interference control to a constraint. Such a remedy results in a nonlinear formulation, which motivates us to propose a greedy algorithm to solve it in an efficient way.
\paragraph{Embedding Sum Rate Maximization and Interference Control}

Given the subset of selected users $\Sc=\{i:y^*_i=1\}$, the achievable rate of $i$-th user with downlink precoder $\pv_i$ can be written by
\begin{align}
R_i=\log\left(1+\frac{\Abs{\hv_i^\H \pv_i}}{\sigma^2+\sum_{j \in \Sc \backslash i} \Abs{\hv_i^\H \pv_j}} \right).
\end{align}
For the sake of tractability of optimization, we consider an asymptotic version of sum rate when $M_x,M_y \to \infty$ so that the asymptotic zero-forcing precoder of $i$-th user can be simply written by the column vectors of common basis $\Vm$. In particular, we have
\begin{align}
\pv_i & \in\Rc\{\hv_i \} \cap \Nc \{\hv_j, j\in \Sc \backslash i\}\\
&={\{\Vm_m: x_m \trace([\Sigmam_{i}]_m) \ge \delta, x_my_j \trace([\Sigmam_{j}]_m) \le \delta, \forall j\in[N_U] \backslash i,m\in [M/2]\}}
\end{align}
where $\Rc\{\cdot\}$ and $ \Nc\{\cdot\}$ are the range and null spaces of the subspace spanned by the vectors, and $\delta$ is a threshold to determine if the block beam is strong enough to be considered. 

Hence, with such asymptotic precoder, the asymptotic rate\footnote{We point out that the asymptotic rate here is with respect to the number of antennas, which is different from those at high SNR in the literature.} of $i$-th user can be written as 
\begin{align}
R_i^\infty&=\log{\left(1+\frac{\trace\left(\sum_{m=1}^{M/2}x_m[\Sigmam_i]_m[\Sigmam_i]_m^\H\right)}{\sigma^2+\trace\left(\sum_{m=1}^{M/2}\sum_{j=1,j\neq i}^{N_U}y_j x_m[\Sigmam_i]_m[\Sigmam_j]_m^\H\right)}\right)}\\
&=\log{\Big(\sigma^2+\trace\big(\sum_{m=1}^{M/2}\sum_{j=1}^{N_U}y_j x_m[\Sigmam_i]_m[\Sigmam_j]_m^\H\big)\Big)}-\log{\Big(\sigma^2+\trace\big(\sum_{m=1}^{M/2}\sum_{j=1,j\neq i}^{N_U}y_j x_m[\Sigmam_i]_m[\Sigmam_j]_m^\H\big)\Big)} \notag \\
&\ge \sum_{m=1}^{M/2} \left(\log{\Big(\frac{2\sigma^2}{M}+\trace\big(\sum_{j=1}^{N_U}y_j x_m[\Sigmam_i]_m[\Sigmam_j]_m^\H\big)\Big)} - \eta_{i,m}\right)
\end{align}
where the first term is due to Jensen's inequality with $\log(\cdot)$ being a concave function, and the second term is due to an artificially introduced constraint
\begin{align}
\log{\Big(\sigma^2+\trace\big(\sum_{m=1}^{M/2}\sum_{j=1,j\neq i}^{N_U}y_j x_m[\Sigmam_i]_m[\Sigmam_j]_m^\H\big)\Big)} \le \sum_{m=1}^{M/2} {\eta_{i,m}}
\end{align}
With Jensen's inequality, the above constraint can be relaxed to
\begin{align} \label{eq:rate-constraint}
\sum_{m=1}^{M/2} \log{\Big(\frac{2\sigma^2}{M}+\trace\big(\sum_{j=1,j\neq i}^{N_U}y_j x_m[\Sigmam_i]_m[\Sigmam_j]_m^\H\big)\Big)} \le \sum_{m=1}^{M/2} \eta_{i,m}.
\end{align}

Let us introduce two matrices $\Pm \in \CC^{N_U \times \frac{M}{2}}$ and $\Cm \in \CC^{N_U \times \frac{M}{2}}$ such that
\begin{align}
[\Pm]_{i,m}&=\log \trace\Big(\sum_{j=1}^{N_U}y_j[\Sigmam_i]_m[\Sigmam_j]_m^\H\Big),\label{profit}\\
[\Cm]_{i,m}&=\log \trace\Big(\sum_{j=1,j\neq i}^{N_U}y_j[\Sigmam_i]_m[\Sigmam_j]_m^\H\Big).\label{cost}
\end{align}
The maximization of the asymptotic sum rate with joint user and beam selection can be approximately formulated in the following way
\begin{subequations}\label{eq:GAP}
\begin{align}
(\Pc_3): \max_{z_{i,m}} \quad & \sum_{m=1}^{M/2} \sum_{i=1}^{N_U} z_{i,m} [\Pm]_{i,m} \label{eq:gap-obj}\\
\st \quad &
\eqref{eq:gap-beam},
\eqref{eq:gap-user},\\
&[\Cm]_{i,m} z_{i,m}  \le \eta_{i,m}, \ \forall b_m\in\Bc, \forall u_i\in\Uc \label{eq:gap-cost} \\
&z_{i,m}\in \{0,1\}, \ \forall b_m\in\Bc, \forall u_i\in\Uc
\end{align}
\end{subequations}
where the objective function \eqref{eq:gap-obj} comes from the lower bound of the asymptotic rate, with the constant parts dropped for simplicity, 
and the final constraint \eqref{eq:gap-cost} due to the constraint \eqref{eq:rate-constraint} to control interference, and $z_{i,m}=x_my_i$ is binary-valued. 

The above optimization formulation can be recognized as a GMAP with an additional constraint \eqref{eq:gap-cost}.
The lower bound of the asymptotic rate can be regraded as the profits, and the constrained term in \eqref{eq:rate-constraint} can be treated as costs. As such, we refer to $\Pm$ and $\Cm$ as the profits and costs matrices, respectively.
While the optimization problem \eqref{eq:GAP} is linear for the parameters $\{z_{i,m}\}$, the profits and costs matrices $\Pm$ and $\Cm$ are dependent of user selection $\{y_j\}$, which is entangled with $\{z_{i,m}\}$ as $z_{i,m}=x_my_i$. This makes the problem a nonlinear integer program with respect to $\{x_m\}$ and $\{y_i\}$, which is challenging to solve. To overcome this, we propose a low-complexity greedy algorithm, avoiding overlaps between any two users in the matrix-weight bipartite graph. 

\paragraph{Greedy Algorithm}
As detailed earlier, given the channel covariance matrix $\hat{\Rm}_i$ with permuted rows and columns from the original one ${\Rm}_i$, we can construct a matrix-weight bipartite graph representation where the matrix weights $[{\Sigmam}_i]_m$ come from the block diagonalization of $\hat{\Rm}_i$.
{However, when it comes to the practical scenarios with a finite number of antennas, $\hat{\Rm}_i$ is not perfectly block-diagonalizable with the DFT matrix as in \eqref{eq:blk-diag}. To overcome this, a possible way is to approximate the matrix-weight $[{\Sigmam}_i]_m$ by
\begin{align} \label{eq:practical-decompose}
[\hat{\Sigmam}_i]_m &= [(\Fm_{M_y}\otimes \Fm_{M_x}\otimes \IM_2)^\H \hat{\Rm}_i (\Fm_{M_y}\otimes \Fm_{M_x}\otimes \IM_2)]_{m,m}
\end{align}
where $[\cdot]_{m,m}$ is the $m$-th $2\times 2$ block diagonal submatrix with $m \in [M/2]$.
It is readily verified that $\lim_{M_x,M_y \to \infty} [\hat{\Sigmam}_i]_m = [{\Sigmam}_i]_m$ for all $i,m$. Thus, in what follows, we use $[\hat{\Sigmam}_i]_m$ instead of $[{\Sigmam}_i]_m$ for algorithm design in the practical scenarios.
}

For ease of presentation, we introduce a $N_U \times \frac{M}{2} $ matrix $\Psim$ with $[\Psim]_{i,m}=\trace([\hat{\Sigmam}_i]_m)$ to indicate the contribution of the $m$-th block-beam to the $i$-th user.
Let us define a binary matrix $\Am'$ for the greedy algorithm with elements specified as    
\begin{align} \label{eq:adj-matrix}
[\Am']_{i,m}=\left\{ \Pmatrix{1, \quad
\text{if } m \in \text{max}^{n_p} \big\{[\Psim]_{i,m'},\forall m' \in [M/2]\big\},\\0, \quad \text{Otherwise,}}
\right.
\end{align}
where $n_p\in[M/2]$ is a tunable integer parameter, and $\text{max}^{n_p}\{\Ac\}$ returns the indices of the largest $n_p$ values in the set $\Ac$.  
Here $\Am'$ serve as a mask to filter out the insignificant weight matrices $\{[\hat{\Sigmam}_i]_m, \forall m\}$ and only keep $n_p$ largest ones. 
In particular, if we set $n_p=\kappa_b$, then after masking with $\Am'$, there are at most $\kappa_b$ block beams left that are connected to each user, so that the constraint \eqref{eq:gap-user} is automatically satisfied.

For the greedy algorithm, according to the asymptotic analysis of sum rate, we define a specific evaluation function as 
\begin{align}\label{eq:final-profit}
\Phi(\Pm,\Cm) = \sum_{i} \sum_m x_my_i([\Pm]_{i,m} - [\Cm]_{i,m}),
\end{align}
where $\Pm$ and $\Cm$ are profit and cost matrix as shown in \eqref{profit} and \eqref{cost}, for which $[\hat{\Sigmam}_i]_m$ is used. 

Given the bipartite graph representation $\Gc=(\Bc,\Uc,\Ec)$ and the matrix-weight $[\hat{\Sigmam}_i]_m$ on the edge $(b_m,u_i)$, we propose a greedy algorithm to solve the optimization problem $(\Pc_3)$.
The detailed procedure is outlined in Algorithm \ref{greedy_algorithm}. 
\begin{algorithm}[t]
\caption{Greedy Algorithm for Generalized MD-ACS}
\begin{algorithmic}[1]
\State {\bf Input}: $\{\hat{\Rm}_i, \forall i\}$, $\kappa_u, \kappa_b$
\State {\bf Initialization}: $x_m=y_i=1$ for all $i \in [N_U],m \in [M/2]$ 
\State Produce $2 \times 2$ diagonal submatrices $\{[\hat{\Sigmam}_i]_m\}_{i=1,m=1}^{N_U, M/2}$ from $\{\hat{\Rm}_i\}_{i=1}^{N_U}$ according to \eqref{eq:practical-decompose} 
\State Construct the bipartite graph representation $\Gc=(\Bc,\Uc,\Ec)$ with matrix-weights $\{[\hat{\Sigmam}_i]_m\}$ for the edge $(b_m,u_i) \in \Ec$, and construct the weight matrix $\Psim$ with $[\Psim]_{i,m}=\trace([\hat{\Sigmam}_i]_m)$
\State Compute profit and cost matrix $\Pm, \Cm$ as \eqref{profit}-\eqref{cost} with $[\hat{\Sigmam}_i]_m$
\State Construct a binary matrix $\Am'$ as in \eqref{eq:adj-matrix} with $n_p=\kappa_b$, such that \eqref{eq:gap-user} is satisfied 
\State Update $\Psim$ as $\Psim\gets\Am'\odot \Psim$, and set $\Mc = \{m: \sum_{i=1}^{N_U} x_my_i [\Am']_{i,m}> \kappa_u, \forall m \in [M/2]\}$
\While{$\Mc \neq \emptyset$}  
\State Select the beam $m \in \Mc$ 
\State Compute \eqref{eq:final-profit} as $\Phi_b$ if the beam is switched off, i.e., $x_m=0$ 
\State Compute \eqref{eq:final-profit} as $\Phi_u$ if only $\kappa_u$ users with the largest $[\Psim]_{i,m}$ are selected, i.e., $y_i=0$ for all $i \notin \max^{\kappa_u}\{[\Psim]_{i',m}, \forall i'\}$
\If{$\Phi_b > \Phi_u$}
\State $x_m = 0$, and $[\Am']_{i,m}=0$, $\forall i \in [N_U]$
\Else
\State $y_i=0$, and $[\Am']_{i,m}=0$, $\forall m \in [M/2]$, $i \notin \max^{\kappa_u}\{[\Psim]_{i',m},\forall i'\}$ 
\EndIf
\State Update $\Psim$ as $\Psim \gets [\Am']\odot \Psim$
\State Update $\Mc \gets \{m: \sum_{i=1}^{N_U} x_m y_i [\Am']_{i,m}> \kappa_u, \forall m \in [M/2]\}$
\EndWhile
\State {\bf Output}: $\{x_m\}_{m=1}^{M/2}$, $\{y_i\}_{i=1}^{N_U}$ 
\end{algorithmic}  
\label{greedy_algorithm} 
\end{algorithm}

Let us explain the greedy algorithm in detail.  
Instead of maximizing the profit with the cost as the constraint in \eqref{eq:GAP}, we define a new profit function as in \eqref{eq:final-profit} which takes both original profits and costs into account.
At the beginning, each user $i$ selects $\kappa_b$ block beams with the largest weights as specified by $\Am'$ in \eqref{eq:adj-matrix}. This is to make the constraint \eqref{eq:gap-user} automatically satisfied. 
Then, each block beam $m$ determines whether the number of served users exceeds its capability $\kappa_u$ to satisfy the constraint \eqref{eq:gap-beam}. 
For those beams with more than $\kappa_u$ users served that violate the constraint \eqref{eq:gap-beam}, we need to determine if it is better to switch off this beam $m$, or some users so that the constraint \eqref{eq:gap-beam} is satisfied. To make such a decision, we compute and compare two quantities $\Phi_b$ and $\Phi_u$ when either option is applied with respect to the newly defined profit in \eqref{eq:final-profit}. This operation repeats till the constraint \eqref{eq:gap-beam} is satisfied for all active block beams. After each iteration, the weight matrix $\Phim$ and the binary matrix $\Am'$ will be updated, so that the deactivated users or beams will not be considered in the future.  As such, the greedy algorithm results in a feasible solution after at most $\frac{M}{2}$ updates.

To summarize, compared with the original single-dimension ACS formulation in \cite{B.K,khalilsarai2020dual}, our proposed MD-ACS with greedy algorithm has the following advantages.

\begin{itemize}
\item While the original ACS is dedicated to the maximization of multiplexing gain, our proposed MD-ACS takes both sum rate maximization and interference control into account, which leads to better performance at finite SNR, as will be shown in Section \ref{sec:numerical}.
\item In the original ACS, the same threshold is applied for all users and beams to construct the bipartite graph representation, and the resulting graph is sensitive to such threshold. In addition, there are quite a few tunable parameters in \eqref{eq:ACS-MILP}, which are challenging to fine-tune to arrive at the sweet spot for the optimal solution, so that an improper choice probably results in severe performance degradation. For our proposed MD-ACS, the integer-valued parameters $\kappa_u$ and $\kappa_b$ are used to construct the bipartite graph, and the resulting graph is more flexible and suitable for greedy search. 
\item  Due to the pre-determined bipartite graph representation and the implicit user selection, the original ACS is suitable to the homogeneous scenarios, whereas our proposed greedy algorithm is suitable for both homogeneous and heterogeneous scenarios (e.g., with both indoor and outdoor users), thanks to the adaptive bipartite graph construction and the explicit user selection, as will be demonstrated in Section \ref{sec:numerical}.  
\end{itemize}

\section{Numerical Results}
\label{sec:numerical}
In this section, we provide the numerical results of our proposed method --- generalized multi-dimensional active channel sparsification (MD-ACS) --- compared with the state-of-the-art ones in the practical DP-UPA FDD massive MIMO scenarios.   
The following baseline methods are considered for comparison.
\begin{itemize}
\item \textbf{No Selection}: All users and beams are activated.
\item \textbf{JSDM}: A clustering algorithm that divides users into groups according to the similarity of their channel covariance matrices \cite{JSDM}. {In each group, a user is randomly selected on behalf of the corresponding cluster}. In JSDM, the number of clusters $K$ is essential, and is set to $K=\sum_i y_i^*$, where $\yv^*$ is the user selection vector obtained by our greedy algorithm.
\item \textbf{ACS}: The original ACS on scale-weight bipartite graph representation, which was first proposed in \cite{B.K} for ULA, and later on extended to DP-ULA in \cite{khalilsarai2020dual};
\item \textbf{ACS-Matrix}: The conventional ACS with the MILP formulation on the matrix-weight bipartite graph representation, where the constraint \eqref{eq:MILP-5} is replaced by 
\begin{align}
 P y_i \le \sum_{b_m \in \Bc} \trace\left([\Wm]_{i,m}\right) x_{m}, \quad \forall  u_i \in \Uc.
\end{align}
\item \textbf{Greedy Algorithm}: The proposed MD-ACS with greedy algorithm implementation as specified in Algorithm \ref{greedy_algorithm}.
\end{itemize}
The downlink channel training and precoding follow the procedure in Section \ref{sec:signal-model}, where the pilot matrix $\Sm$ is a $T\times M'$ orthogonal matrix, with {$M'=2\sum_{m=1}^{M/2}x_m \le M$ being the number of activated virtual beams}, and the average pilot signal power is set to $\rho^{\text{p}}=1$. 

\begin{figure}
\begin{minipage}{0.48\linewidth}
\centering
\includegraphics[width= 3.2in,angle=0]{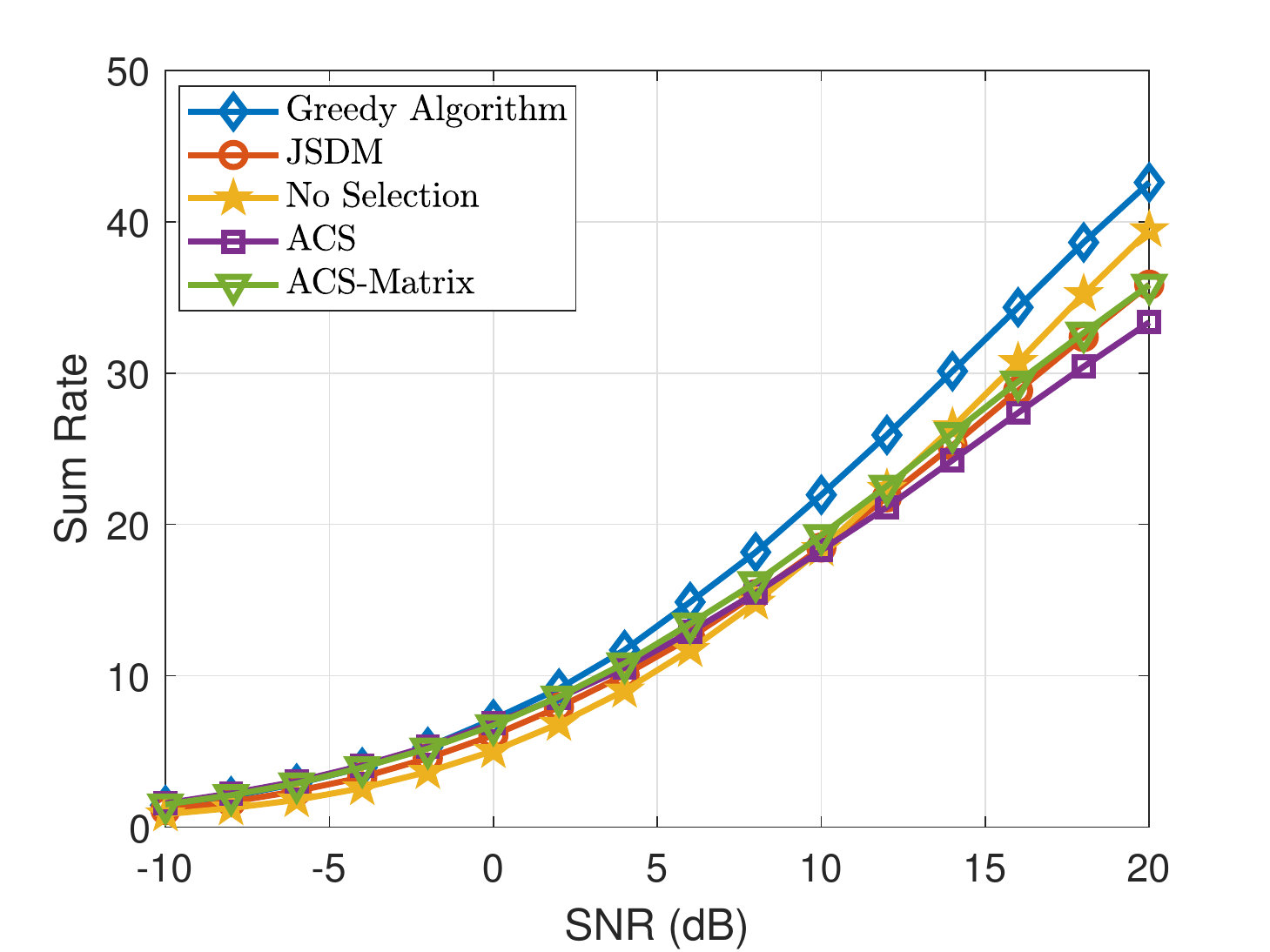}
\caption{Sum rate versus SNR with $4\times 4\times 2$ DP-UPA, $N_U=15$ users and $T=16$ timeslots.}
\label{SNR_32_15}
\end{minipage}
\quad
\begin{minipage}{0.48\linewidth}
\centering
\includegraphics[width= 3.2in,angle=0]{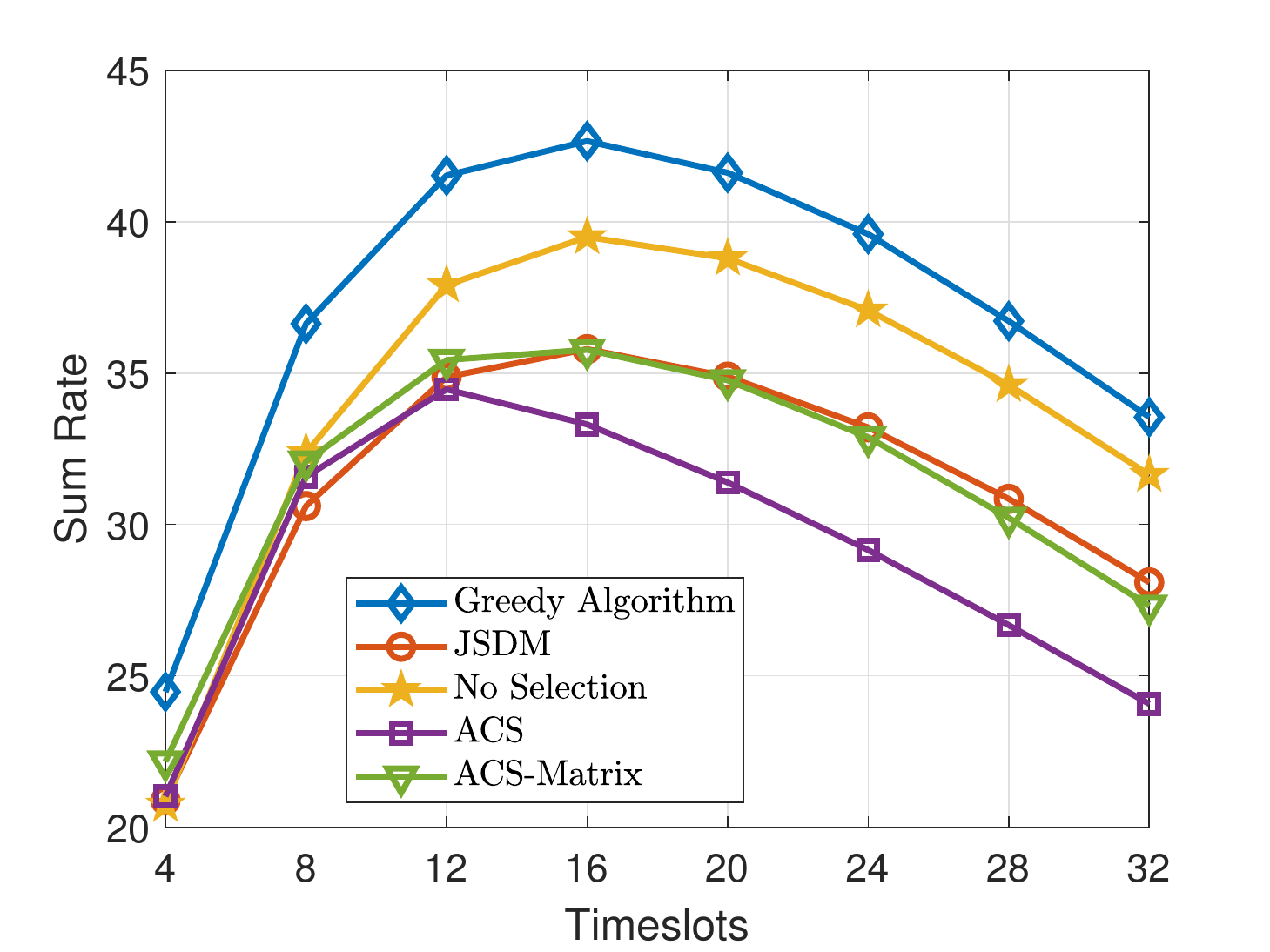}
\caption{Sum rate versus pilot dimension with $4\times 4\times 2$ DP-UPA, $N_U=15$ users and $\mathrm{SNR}=20$ dB.}
\label{T_32_15}
\end{minipage}
\begin{minipage}{0.48\linewidth}
\centering
\includegraphics[width= 3.2in,angle=0]{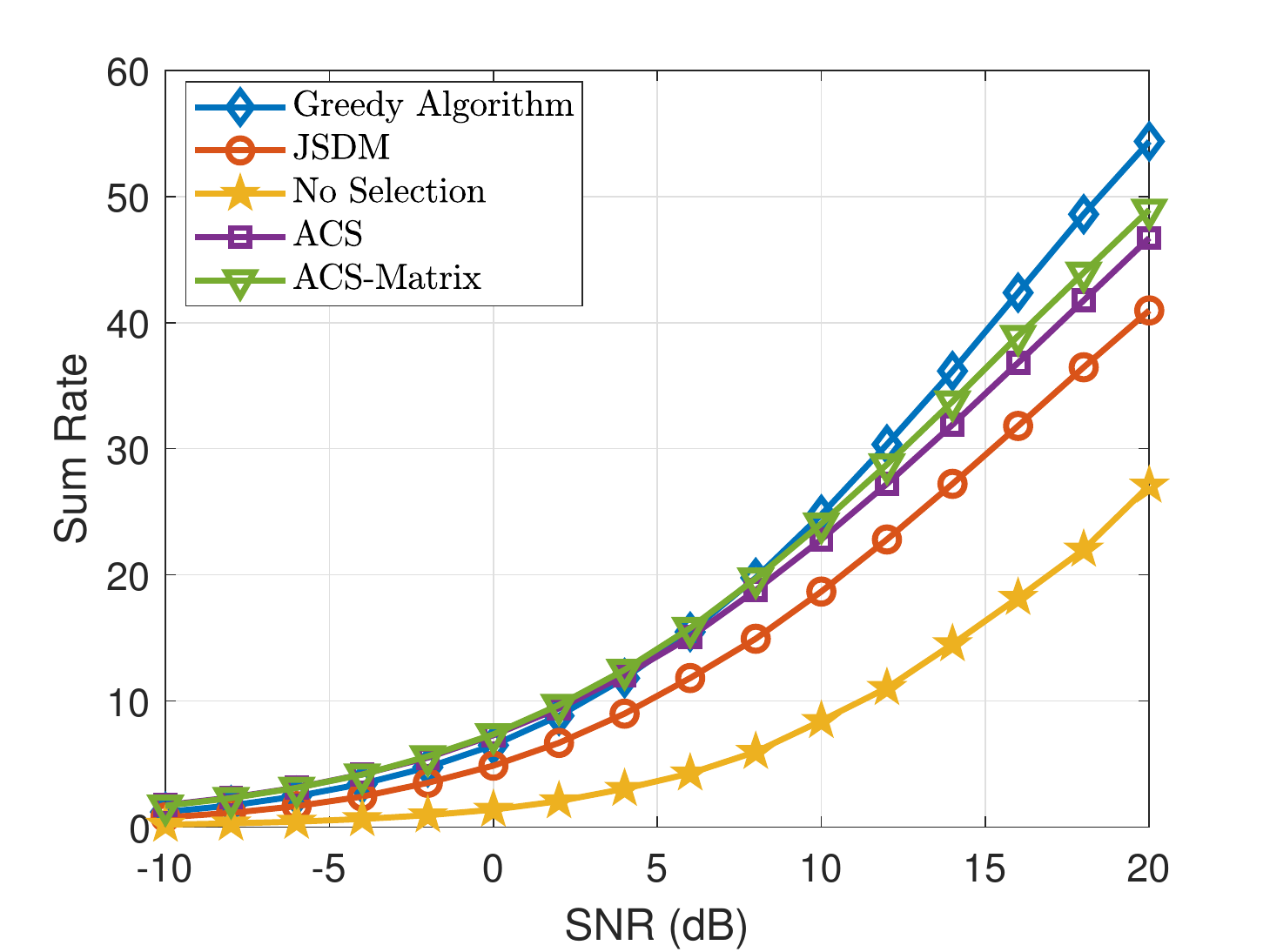}
\caption{Sum rate versus SNR with $4\times 4\times 2$ DP-UPA, $N_U=30$ users and $T=16$ timeslots.}
\label{SNR_32_30}
\end{minipage}
\quad
\begin{minipage}{0.48\linewidth}
\centering
\includegraphics[width= 3.2in,angle=0]{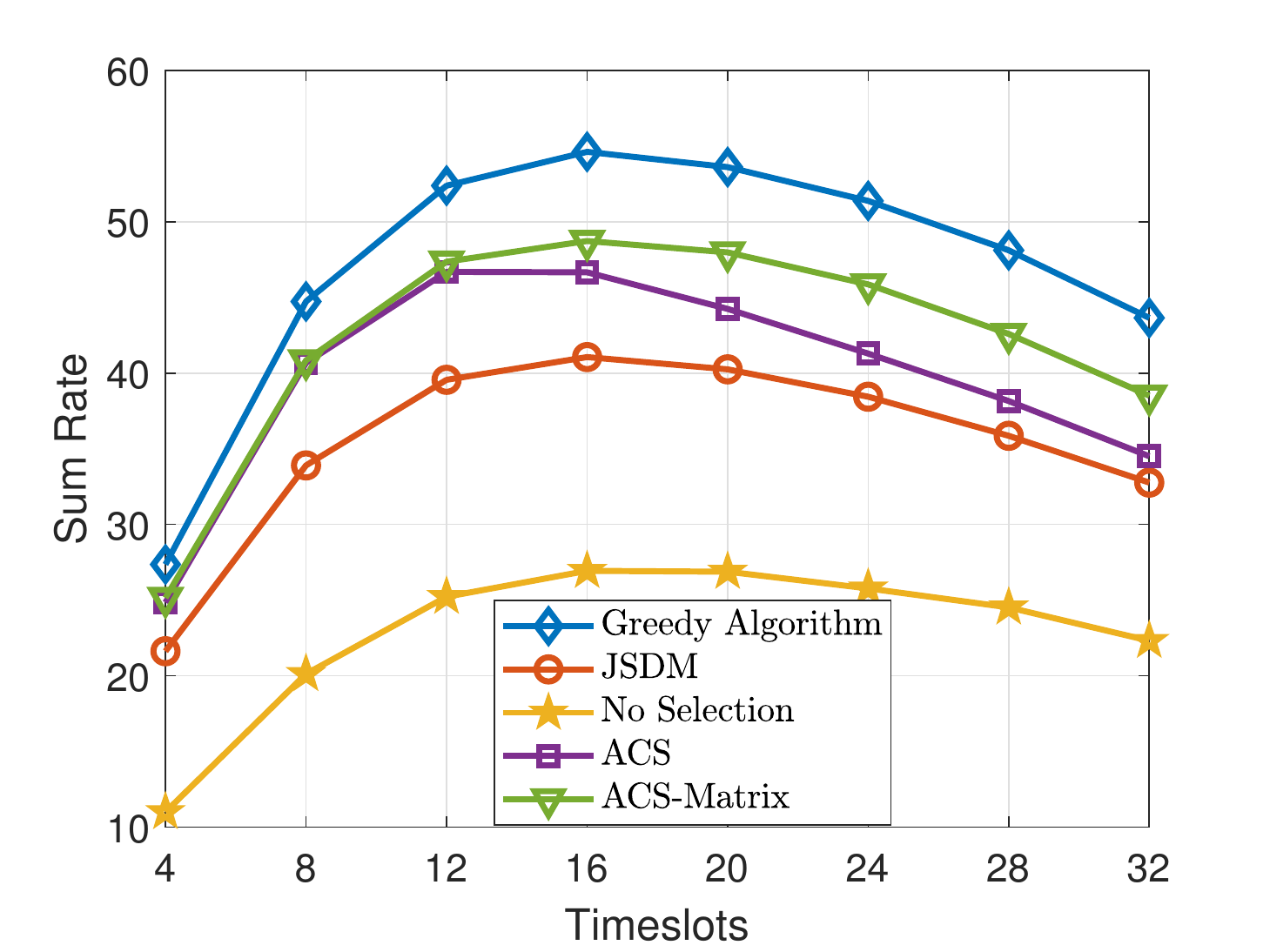}
\caption{Sum rate versus pilot dimension with $4\times 4\times 2$ DP-UPA, $N_U=30$ users and $\mathrm{SNR}=20$ dB.}
\label{T_32_30}
\end{minipage}
\vspace{-20pt}
\end{figure}

For the simulation scenarios, we consider FDD downlink transmission in a single-cell massive MIMO network, where the base station is equipped with $M=M_x \times M_y \times 2$ DP-UPA antenna and serves $N_U$ single-antenna users. In order to evaluate the algorithms comprehensively and fairly, we adopt the QuaDRiGa channel model \cite{jaeckel2014quadriga} to generate downlink channel vectors $\hv$. According to the 3GPP and the QuaDRiGa specifications \cite{QuadrigaD}, the Inter-Site Distance is set to be 500m and the `3GPP-3D-UMA' scenario is considered. For the users' located rules, the minimum distance from users to the base station is 10m. We choose 50\% indoor and 50\% outdoor users for downlink transmission, where the height of all the outdoor users is set to 1.5m. In all simulation scenarios, we assume the downlink channel covariance matrix is somehow available, which can be simply obtained by $\Rm=\frac{1}{N}\sum_{t=1}^{N} \hv_t \hv_t^\H$ using $N=1000$ downlink channel vector realizations $\hv_t$ generated from QuaDRiGa, or obtained from uplink channel covariance matrix by leveraging uplink-downlink reciprocity (e.g., \cite{B.K}). Unless otherwise explicitly specified, for all the simulation scenarios, we choose the following parameters: $\kappa_b=3$; $\kappa_u=20$ for 64 antenna configuration and $\kappa_u=12$ for 32 antenna configuration; pilot dimensions of training phase $T=16$; for 32 antennas (Fig. \ref{SNR_32_15}-\ref{T_32_30}), the antenna array is $M_x=M_y=4$ and the FDD frame length is $T_c = 64$, and for 64 antennas (Fig. \ref{SNR_64_30}-\ref{User_64_20T}), $M_x=4$, $M_y=8$ and $T_c = 72$.

Figures \ref{SNR_32_15} and \ref{T_32_15} illustrate the sum rate of downlink transmission with $N_U=15$ users in total versus SNR and the pilot dimension of the training phase, respectively. We can observe in Fig. \ref{SNR_32_15} that, the proposed MD-ACS with greedy algorithm outperforms all other methods, and the gap is increasing as SNR goes. The ACS-like methods (i.e., ACS and ACS-Matrix) perform poorly at high SNR compared with No Selection, which is probably due to the fact that the number of users for selection is quite limited so that activating all users may not be a bad idea.
In Fig. \ref{T_32_15}, it appears the sum rate first increases as the pilot dimension does, because higher pilot dimension yields higher estimation accuracy of downlink channel, and therefore more accuracy downlink precoding. The sum rate is decreasing when pilot dimension increases further, because the more resource the training phase occupies, the less the transmission phase could use.
For the ACS-like methods (i.e., ACS and ACS-Matrix), it looks too many users have been switched off, which results in severe performance degradation when $T$ is large. 
In Figures \ref{SNR_32_30} and \ref{T_32_30}, $N_U=30$ users are considered. It is observed that our proposed MD-ACS with greedy algorithm consistently outperform other methods. In this scenario, with a sufficiently large number of users, both ACS-Matrix and ACS perform better than JSDM and No Selection. The No Selection method confronts severe performance degradation -- it is because there are too many users in the network, and user selection is crucial. In these simulations, ACS-Matrix outperforms the conventional ACS approach, which demonstrates the effectiveness of using matrix-weight bipartite graph representation.
Notably, from Fig. \ref{T_32_15} and Fig. \ref{T_32_30}, the optimal pilot dimensions that maximize the sum rate are different across algorithms. The optimal pilot dimension of ACS is around $T=12$ while others are around $T = 16$. This suggests that ACS seems more dedicated to beam selection, while others (including the greedy algorithm) prefer user selection. 

\begin{figure}
\begin{minipage}{0.48\linewidth}
\centering
\includegraphics[width= 3.2in,angle=0]{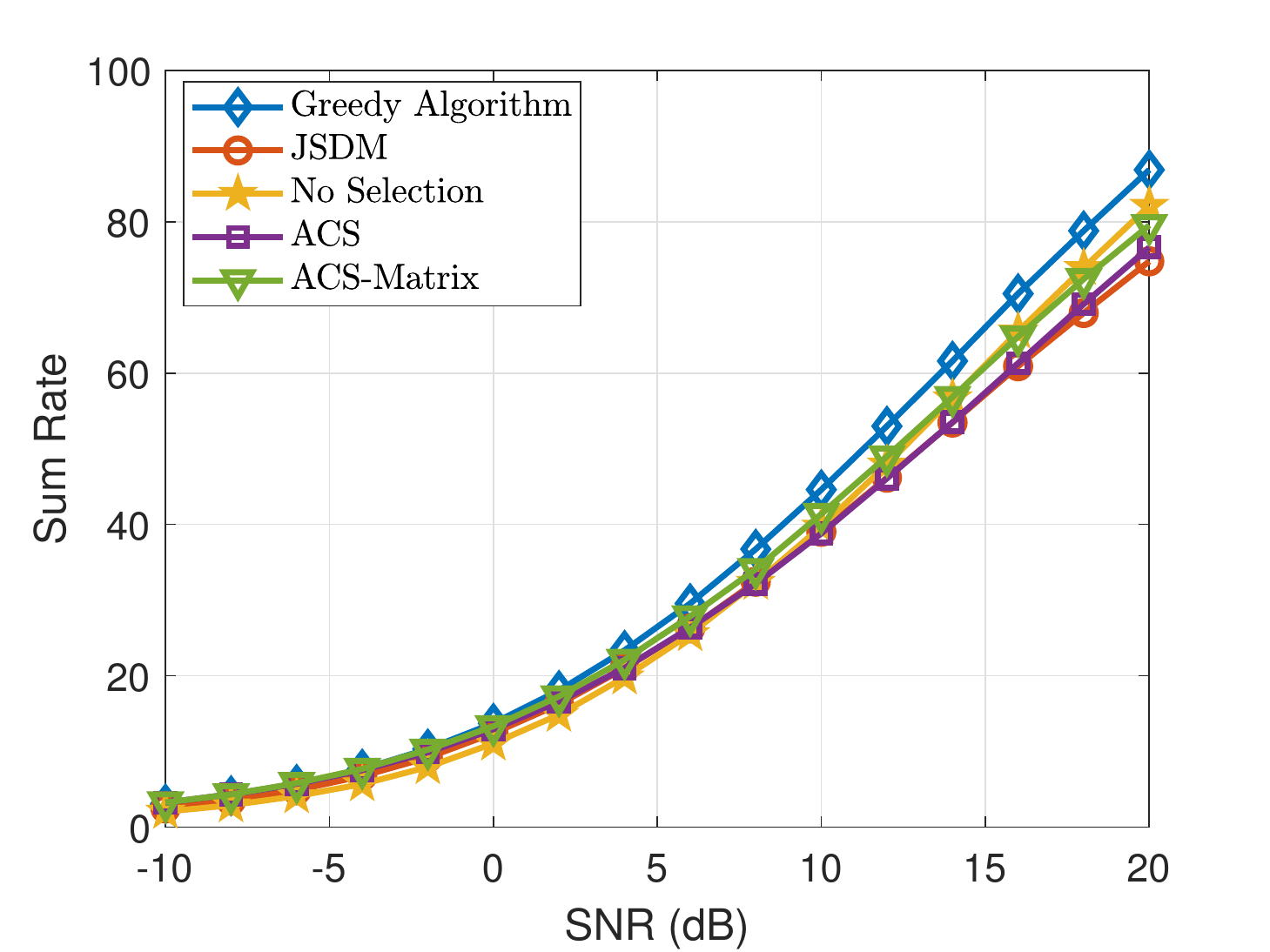}
\caption{Sum Rate versus SNR with $4\times 8\times 2$ DP-UPA, $N_U=30$ users and $T=16$ timeslots.}
\label{SNR_64_30}
\end{minipage}
\quad
\begin{minipage}{0.48\linewidth}
\centering
\includegraphics[width= 3.2in,angle=0]{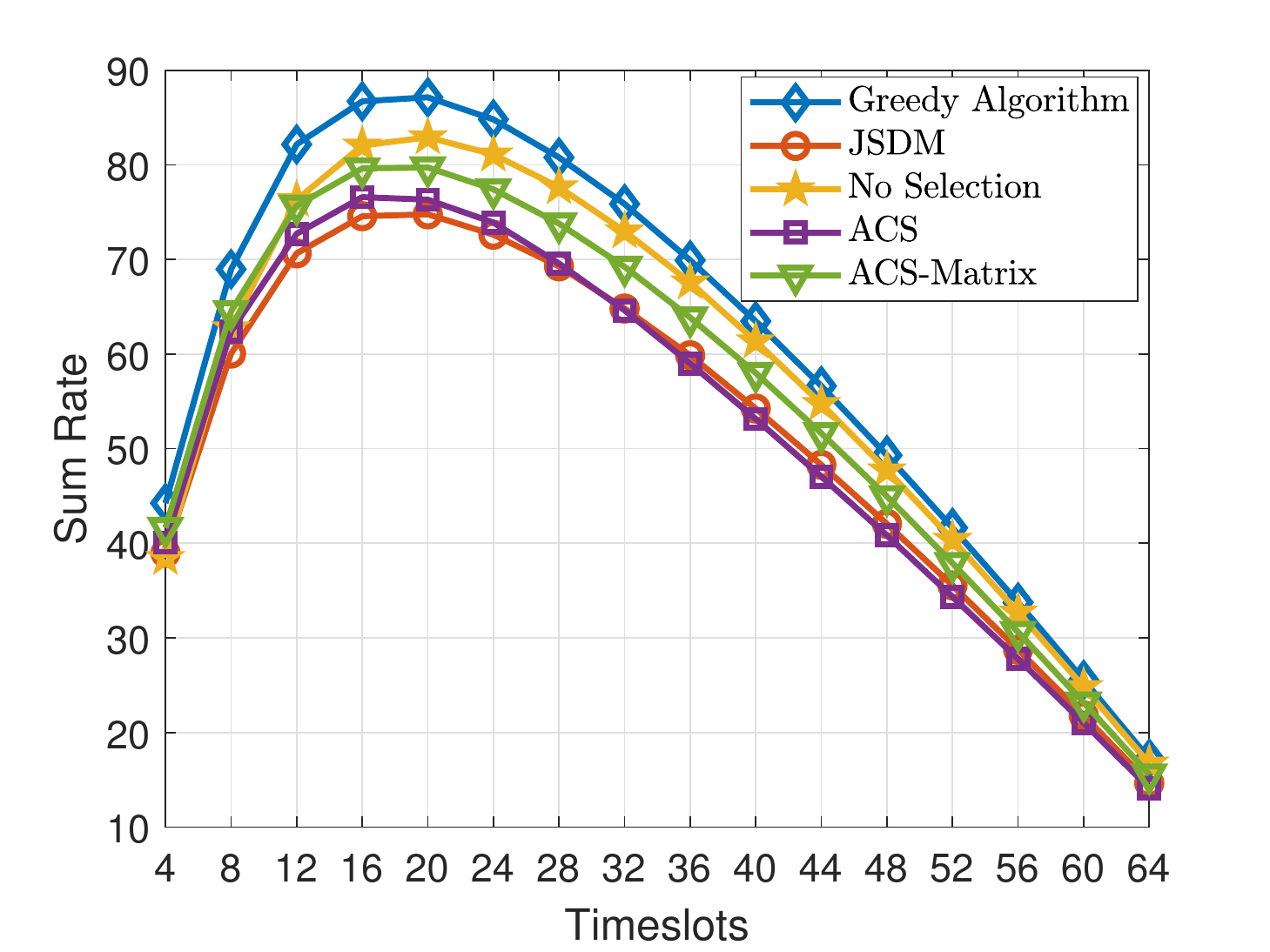}
\caption{Sum Rate versus Timeslots with $4\times 8\times 2$ DP-UPA, $N_U=30$ users and $\mathrm{SNR}=20$ dB.}
\label{T_64_30}
\end{minipage}
\begin{minipage}{0.48\linewidth}
\centering
\includegraphics[width= 3.2in,angle=0]{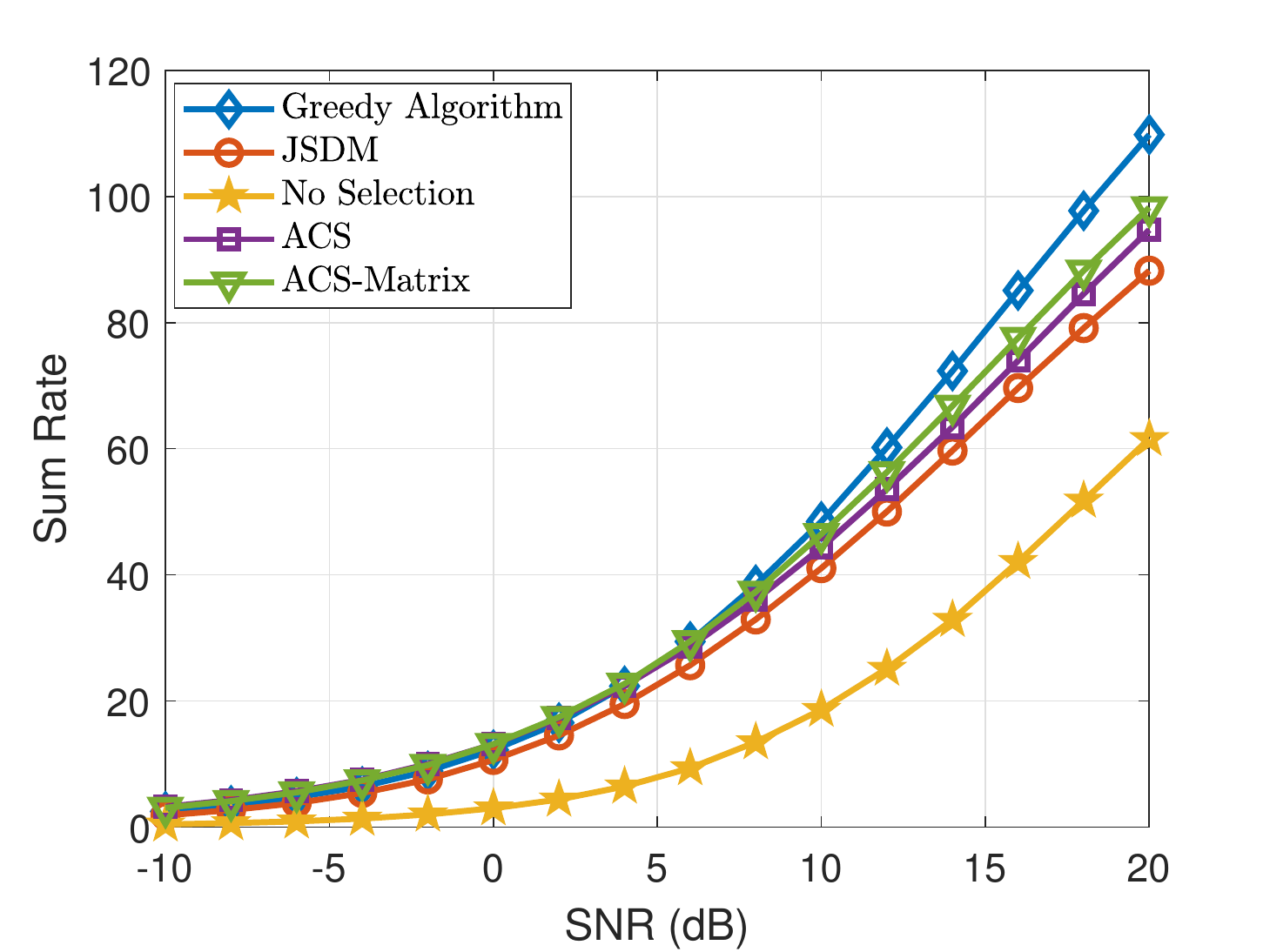}
\caption{Sum Rate versus SNR with  $4\times 8\times 2$ DP-UPA, $N_U=60$ users and $T=16$ timeslots.}
\label{SNR_64_60}
\end{minipage}
\quad
\begin{minipage}{0.48\linewidth}
\centering
\includegraphics[width= 3.2in,angle=0]{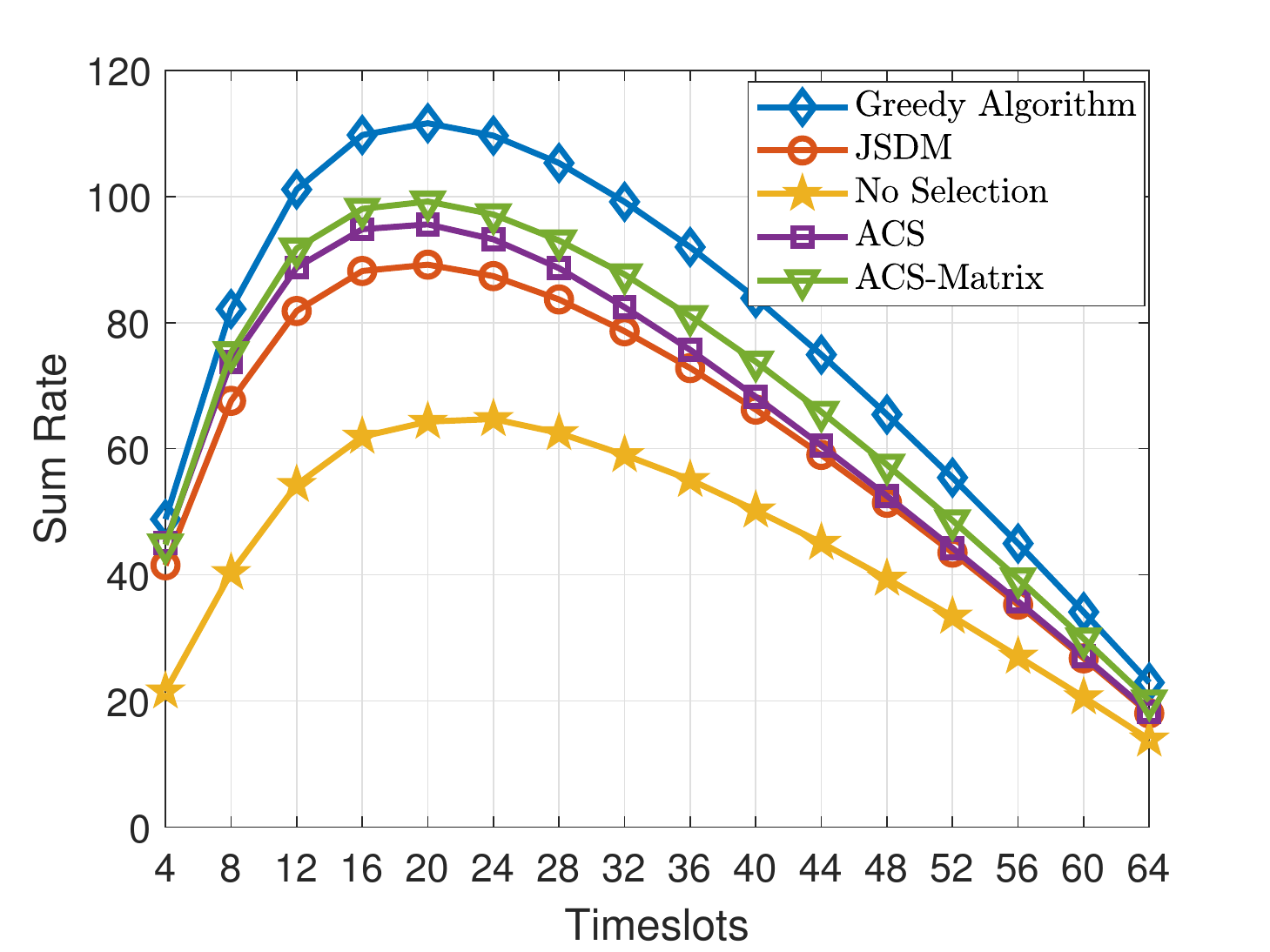}
\caption{Sum Rate versus Timeslots with $4\times 8\times 2$ DP-UPA, $N_U=60$ users and $\mathrm{SNR}=20$ dB.}
\label{T_64_60}
\end{minipage}
\end{figure}

Further, in Fig. \ref{SNR_64_30}-\ref{T_64_60}, we increase the number of antennas from 32 to 64, and consider two scenarios with $N_U=30$ and $N_U=60$ users. 
For the 64 DP-UPA antenna scenario with 30 users, compared with the 32-antenna case, the sum rate improvement of our greedy algorithm over other methods is diminishing, because the capability of serving more users is enhanced with more antennas, and thus user selection is not crucial.
When the number of users increases from $N_U=30$ to $N_U=60$, we observe the same phenomenon as that in Figures \ref{SNR_32_30} and \ref{T_32_30}. Interestingly, from Figures \ref{T_64_30} and \ref{T_64_60}, we find that even if the number of antennas is increased, the optimal pilot dimension is still around $T=16$ timeslots.
It is worth noting that, ACS-Matrix with matrix-weight graph representation outperforms the conventional scalar-weight ACS method. It suggests that the matrix-weight formulation is more suitable than scalar-weight ACS for the DP-UPA scenario. As such, the improvement of our proposed greedy algorithm comes from two aspects: the matrix-weight MILP formulation and the search-based user/beam selection strategy. 
\begin{figure}
\centering
\includegraphics[width= 3.2in,angle=0]{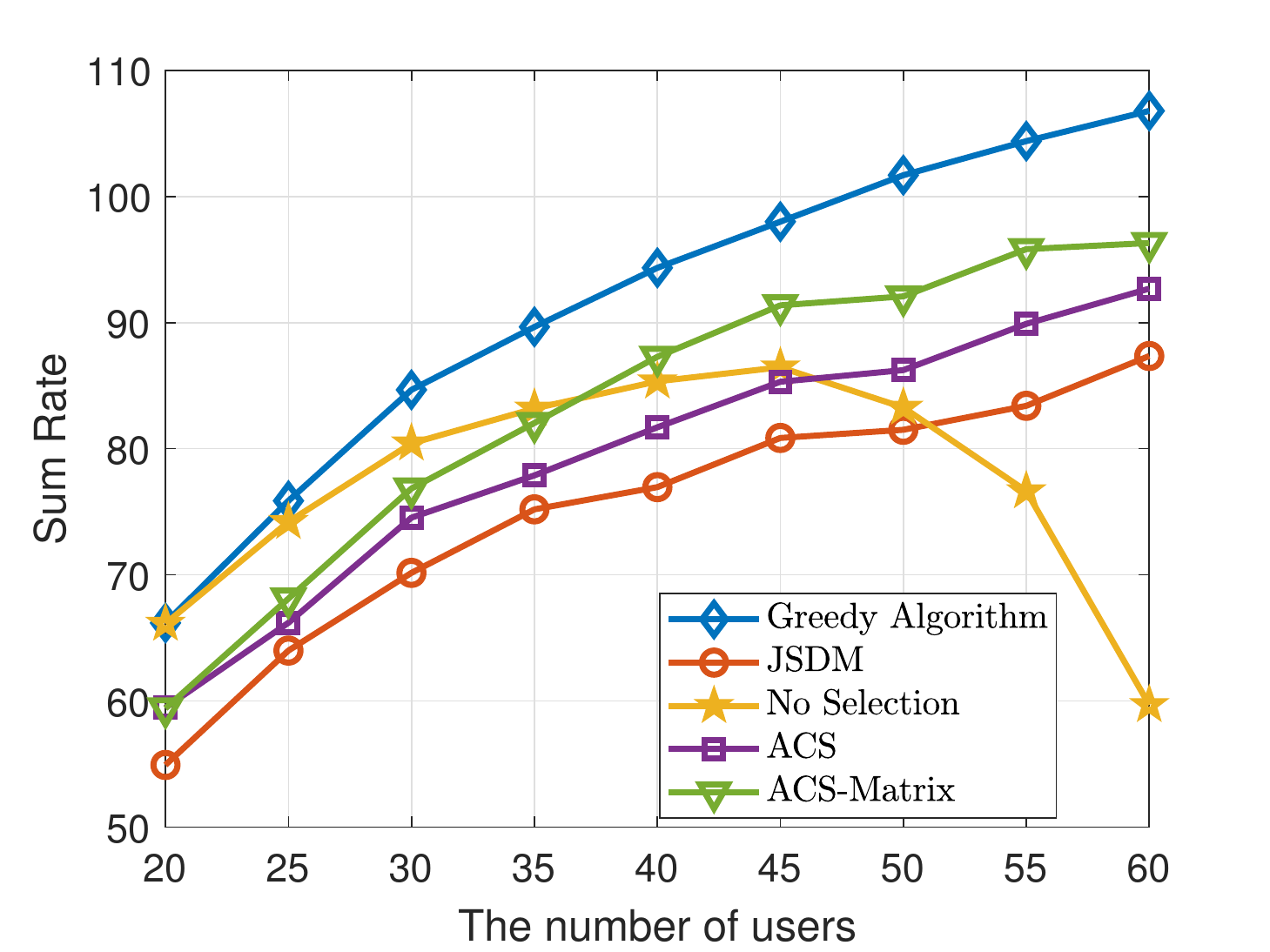}
\caption{Sum Rate versus the number of users with $4\times 8\times 2$ DP-UPA, $T=16$ timeslots and $\mathrm{SNR}=20$ dB.}
\label{User_64_20T}
\vspace{-20pt}
\end{figure}

 Fig. \ref{User_64_20T} shows the sum rate versus the number of users in the cell when the pilot dimension is set to $T=16$. We can observe that: (1) When the number of users is small, e.g., $N_U \le 30$, user selection is unnecessary, because the sum rates are nearly the same for the proposed greedy algorithm, compared with No Selection. (2) As the number of users increases, the benefit of user selection emerges, and it becomes crucial when the number of users is large, e.g., $N_U \ge 45$. (3) Our proposed MD-ACS with greedy algorithm always outperforms the conventional ACS thanks to the matrix-weight graph representation and the search-based greedy user/beam selection.

\section{Conclusion}
In this work, we have investigated downlink sparsifying precoder design and user selection in DP-UPA FDD massive MIMO systems using active channel sparsification (ACS). By extending the original scalar-weight bipartite graph representation of user-beam association
to a matrix-weight bipartite graph, we proposed a generalized multi-dimensional ACS (MD-ACS) for DP-UPA antenna configurations with a nonlinear integer program formulation. Inspired by the generalized multi-assignment problem, we proposed an efficient greedy algorithm to solve the nonlinear integer problem, and observed its superiority in extensive simulation results using QuaDriGa channel models. We believe such an improvement of the ACS methodology could pave the way for the potential deployment of ACS to the practical FDD massive MIMO systems.

 \appendix
\subsection{Proof of Lemma \ref{lemma:structure_cov}}
\label{proof:lemma-toeplitz}

Given the channel vector in \eqref{eq:channelV_channelH}, the covariance matrix $\Rm=\EE\{\hv\hv^\H\}$ can be written as 
\begin{align}
\Rm=
\begin{bmatrix}
\EE\{\hv_V\hv_V^\H\} & \EE\{\hv_V\hv_H^\H\}\\
\EE\{\hv_H\hv_V^\H\} &\EE\{\hv_H\hv_H^\H\}
\end{bmatrix}
\defeq
\begin{bmatrix}
\Rm_1 &\Rm_2\\
\Rm_2^\H &\Rm_3
\end{bmatrix}.
\end{align}
For $k=1,2,3$ we have
\begin{align}
\Rm_k=\int_\Omega {p_k\av(\theta,\phi)\av^{\H}(\theta,\phi)} d\theta d\phi
\end{align}
where $\Omega = \{(\theta, \phi): \theta\in \Ac,\phi\in \Bc\}$ and
\begin{align}
p_1=\gamma_V{\gamma}^*_V\EE\{\beta_V{\beta}^*_V\},\quad p_2=\gamma_V{\gamma}^*_H\EE\{\beta_V\bar{\beta}_H\}, \quad p_3 = \gamma_H{\gamma}^*_H\EE\{\beta_H{\beta}^*_H\}
\end{align}
with $\beta_V$ and $\beta_H$ being vertical and horizontal polarization respectively. In fact, the submatrix $\Rm_k$ has the same structure as the covariance matrix of UPA, which is a doubly Toeplitz matrix. 
By letting $\Pm_{\beta}=\left(\begin{smallmatrix}
p_1 & p_2^\herm\\
p_2 & p_3
\end{smallmatrix}\right)$, the covariance matrix $\Rm$ can be alternatively written as
\begin{align}
\Rm=
\Pm_{\beta}
\otimes \int_\Omega  {\av(\theta,\phi)\av^{\H}(\theta,\phi)} d\theta d\phi
=\Pm_{\beta} \otimes
\Bm
\end{align}
where
\begin{align}
\Bm &=\int_\Omega\av_y(\theta,\phi)\av_y^\H(\theta,\phi) \otimes \av_x(\theta,\phi) \av_x^\H(\theta,\phi) d\theta d\phi\\
&= \int_\Omega
\begin{bmatrix}
\Bm_{11} & \cdots & \Bm_{1M_y}\\
\vdots &\ddots & \vdots \\
\Bm_{M_y1} & \cdots & \Bm_{M_yM_y}\\
\end{bmatrix}
d\theta d\phi.  \label{eq:R-toeplitz}
\end{align}
For $p,q \in [M_y]$, each block $\Bm_{pq}$ can be written as 
\begin{align}
 \Bm_{pq}&= \left[\av_y(\theta,\phi)\av_y^\H(\theta,\phi)\right]_{pq} \Am(\theta,\phi) 
=e^{\jmath \frac{2\pi}{\lambda_w} d_y(p-q)\sin\phi \sin\theta}\Am(\theta,\phi) 
\end{align}
with
\begin{align}
[\Am(\theta,\phi)]_{ij}&=[\av_x(\theta,\phi)]_i[\av_x^\H(\theta,\phi)]_j
=e^{\jmath \frac{2\pi}{\lambda_w} d_x(i-j)\sin\phi\cos\theta}. 
\end{align}
It appears that the elements in $\Am(\theta,\phi)$ only depend on $(i-j)$ and the submatrices in $\Bm_{pq}$ only depend on $(p-q)$. Therefore, we conclude that $\Bm$ is a doubly Toeplitz matrix.

To facilitate the inspection from the perspective of generating function for Toeplitz matrices, we transform $\Rm$ into a doubly block Toeplitz matrix by row/column permutation. 
Following the footsteps in \cite{H.V}, we permutate $\Rm$ by a perfect shuffle matrix $\Qm$ as 
\begin{align}
\hat{\Rm} = \Qm \Rm \Qm^{\T} = \Qm \Pm_{\beta} \otimes \Bm \Qm = \Bm \otimes \Pm_{\beta}
\end{align}
with 
\begin{align}
\Qm=
\begin{bmatrix}
\IM_M\left (1:\frac{M}{2}:M,:\right )\\
\IM_M\left (2:\frac{M}{2}:M,:\right )\\
\vdots\\
\IM_M\left (\frac{M}{2}:\frac{M}{2}:M,:\right )\\
\end{bmatrix}.
\end{align}
By the permutation, $\hat{\Rm}$ is doubly Toeplitz matrix, that is,
 an $M_xM_y \times M_xM_y$ doubly block Toeplitz matrix, with each element being a $2\times 2$ matrix. In particular, the $(m_1,m_2)$-th submatrix $[\hat{\Rm}]_{m_1,m_2}$ can be given by
\begin{align}\label{eq:R_m1,m2}
[\hat{\Rm}]_{m_1,m_2} = \int_\Omega e^{\jmath \frac{2\pi}{\lambda_w} (d_ym_1\sin\phi \sin\theta+ d_xm_2\sin\phi\cos\theta)} d \theta d \phi \Pm_{\beta}.
\end{align}
When $M_x, M_y \to \infty$, it is known in \cite{B.T} that the $(m_1,m_2)$-th submatrix of $\hat{\Rm}$ can be given by
\begin{align}
[\hat{\Rm}]_{m_1,m_2} = \int_{-1/2}^{1/2} \int_{-1/2}^{1/2} \Sigmam(\omega_1,\omega_2) e^{-\jmath 2\pi (m_1\omega_1+m_2\omega_2)} d \omega_1 d \omega_2
\end{align}
through its generating function
\begin{align}
&\Sigmam(\omega_1,\omega_2)=\sum_{m_1=-\infty}^{\infty}\sum_{m_2=-\infty}^{\infty}[\hat{\Rm}]_{m_1,m_2}e^{\jmath 2\pi (m_1\omega_1+m_2\omega_2)}.
\end{align}

It is known that for any Toeplitz matrix $\Tm_n$, when $n \to \infty$, there exists a circulant matrix $\Cm_n$ sharing the same generating function \cite{T.M}. This applies to the extensions, e.g., doubly (block) Toeplitz and circulant matrices. It is known in \cite{B.T} that circulant matrix can be diagonalized by DFT matrix, and this can be extended to block and doubly block Toeplitz matrices. As such, for the doubly block Toeplitz matrix $\hat{\Rm}$, there exists a doubly block circulant matrix $\hat{\Cm}$ such that
\begin{align}
\hat{\Cm}=(\Fm_{M_x}\otimes \Fm_{M_y}\otimes \IM_2) \Sigmam (\Fm_{M_x}\otimes \Fm_{M_y}\otimes \IM_2)^\H
\end{align}
where $\Sigmam$ is a block diagonal matrix with $M_xM_y$ non-zero diagonal blocks of size $2 \times 2$ each. According to \cite[Theorem 2]{yi2020asymptotic}. the diagonal blocks of $\Sigmam$ is the uniform sampling of the generating function $\Sigmam(\omega_1,\omega_2)$ on the following grids
\begin{align}
(\omega_1,\omega_2) = \left(-\frac{1}{2}+\frac{ m_1}{M_y},-\frac{1}{2}+\frac{m_2}{M_x}\right), \quad \forall m_1 \in [M_y]-1, \; m_2 \in [M_x]-1.
\end{align}

 \subsection{Proof of Theorem \ref{theorem:sparsity}}\label{proof:theorem-sparsity}
According to Lemma 1, by letting $d_x=d_y=d$ and plugging \eqref{eq:R_m1,m2} into the spectral density function $\Sigmam(\omega_1,\omega_2)$, we have
\begin{align}
\Sigmam(\omega_1,\omega_2)&=\sum_{m_1=-\infty}^{\infty}\sum_{m_2=-\infty}^{\infty}[\hat{\Rm}]_{m_1,m_2}e^{\jmath2\pi (m_1\omega_1+m_2\omega_2)}\\
&=\Pm_{\beta}
\int_\Omega \sum_{m_1=-\infty}^{\infty}\sum_{m_2=-\infty}^{\infty}e^{\jmath 2\pi m_1(\frac{d}{\lambda_w}\sin\phi\sin\theta+\omega_1)}e^{\jmath 2\pi m_2 (\frac{d}{\lambda_w}\sin\phi\cos\theta+\omega_2)}d\theta d\phi\\
&=\Pm_{\beta}
\int_\Omega \left(\sum_{m_1=-\infty}^{\infty}e^{\jmath 2\pi m_1(\frac{d}{\lambda_w}\sin\phi\sin\theta+\omega_1)}\right) \left(\sum_{m_2=-\infty}^{\infty}e^{\jmath 2\pi m_2 (\frac{d}{\lambda_w}\sin\phi\cos\theta+\omega_2)}\right) d\theta d\phi\\
&=\Pm_{\beta}
\int_\Omega \left(\sum_{m_1=-\infty}^{\infty}\delta\left(m_1-\left(\frac{d}{\lambda_w}\sin\phi\sin \theta+\omega_1\right)\right) \right) \notag \\
& \hspace{4.5cm} \cdot \left(\sum_{m_2=-\infty}^{\infty}\delta\left(m_2-\left(\frac{d}{\lambda_w}\sin\phi\cos\theta+\omega_2\right)\right) \right)d\theta d\phi \label{eq:spectral-func-integral}
\end{align}
where the last equation is due to Poisson Summation  Formula \cite{lapidoth2017foundation}. 

Further, let $z_1=\sin\phi\sin\theta$ and $z_2=\sin\phi\cos \theta $. Define $z_i^{\max}=\max_{\phi, \theta}\{z_i\}$ and $z_i^{\min}=\min_{\phi, \theta}\{z_i\}$. Due to the property of delta function, only if we have both $\omega_1=m_1-\frac{d}{\lambda_w}z_1$ and $\omega_2=m_2-\frac{d}{\lambda_w}z_2$, $\Sigmam(\omega_1,\omega_2)$ is a non-zero matrix. Given that $m_1,m_2\in\ZZ$, $-1 \le z_i^{\min} \le z_i^{\max} \le 1$, and $\omega_1, \omega_2 \in \left(-\frac{1}{2},\frac{1}{2}\right)$, the only possible integer of $m_1$ and $m_2$ is 0. Thus, the range of $\omega_i$ that yields non-zero $\Sigmam(\omega_1,\omega_2)$ depends on that of $z_i^{\min}$ and  $z_i^{\max}$, i.e. $\omega_i\in\left[-\frac{d}{\lambda_w}z_i^{\max},\frac{d}{\lambda_w}z_i^{\min}\right], i=\{1,2\}$. 
As such, given a set of AOA $\theta_c, \phi_c$ and AS $\Delta_1, \Delta_2$, we are able to obtain a compact support that is related to the both elevation and azimuth AOAs. 
Even when the special points, $\omega_1, \omega_2 = \pm \frac{1}{2}$, are considered, such that $m_1,m_2=\pm 1$ might exist, we only have the corresponding points under $z_i=\pm1$ that does not alter the conclusion.


\end{document}

%% file: my_macros.tex
\usepackage{amsthm,amsmath,amssymb}
\usepackage{eucal}
\usepackage{xifthen}
\usepackage{mathtools}
\usepackage{enumerate}
\usepackage{microtype}
\usepackage{xspace}
\usepackage{bm}
\usepackage{cite}
\usepackage{fancyhdr}
\usepackage{lastpage}
\usepackage[small]{caption}
\usepackage{xcolor}
\usepackage{algorithm}
\usepackage{algpseudocode}
\usepackage{booktabs}
\usepackage{multirow}
\usepackage{subfigure}
\allowdisplaybreaks

\long\def\comment#1{}

\newfont{\bbb}{msbm10 scaled 700}

\newfont{\bb}{msbm10 scaled 1100}
\newcommand{\CC}{\mbox{\bb C}}

\newcommand{\RR}{\mbox{\bb R}}

\newcommand{\ZZ}{\mbox{\bb Z}}

\newcommand{\EE}{\mbox{\bb E}}

\newcommand{\mbs}[1]{\bm{#1}}
\newcommand{\vect}[1]{{\lowercase{\mbs{#1}}}}
\newcommand{\mat}[1]{{\uppercase{\mbs{#1}}}}

\newcommand{\Pmatrix}[1]{\begin{array}{ll}#1\end{array}}

\newcommand{\T}{{\scriptscriptstyle\mathsf{T}}}
\renewcommand{\H}{{\scriptscriptstyle\mathsf{H}}}

\renewcommand{\Re}[1][]{\ifthenelse{\isempty{#1}}{\operatorname{Re}}{\operatorname{Re}\left(#1\right)}}
\renewcommand{\Im}[1][]{\ifthenelse{\isempty{#1}}{\operatorname{Im}}{\operatorname{Im}\left(#1\right)}}

\newcommand{\av}{\vect{a}}

\newcommand{\fv}{\vect{f}}

\newcommand{\hv}{\vect{h}}

\newcommand{\nv}{\vect{n}}

\newcommand{\pv}{\vect{p}}

\newcommand{\xv}{\vect{x}}
\newcommand{\yv}{\vect{y}}

\newcommand{\zerov}{\vect{0}}

\newcommand{\iotav}{\hbox{\boldmath$\iota$}}

\newcommand{\Lambdam}{\hbox{\boldmath$\Lambda$}}

\newcommand{\Sigmam}{\hbox{\boldmath$\Sigma$}}
\newcommand{\Phim}{\hbox{\boldmath$\Phi$}}

\newcommand{\Psim}{\hbox{\boldmath$\Psi$}}

\newcommand{\Am}{\mat{a}}
\newcommand{\Bm}{\mat{b}}
\newcommand{\Cm}{\mat{c}}

\newcommand{\Em}{\mat{e}}
\newcommand{\Fm}{\mat{f}}

\newcommand{\Pm}{\mat{p}}
\newcommand{\Qm}{\mat{q}}
\newcommand{\Rm}{\mat{r}}
\newcommand{\Sm}{\mat{s}}
\newcommand{\Tm}{\mat{t}}

\newcommand{\Vm}{\mat{V}}
\newcommand{\Wm}{\mat{w}}
\newcommand{\Xm}{\mat{x}}
\newcommand{\Ym}{\mat{y}}

\newcommand{\IM}{\mat{i}}

\newcommand{\Ac}{{\mathcal A}}
\newcommand{\Bc}{{\mathcal B}}

\newcommand{\Ec}{{\mathcal E}}

\newcommand{\Gc}{{\mathcal G}}

\newcommand{\Mc}{{\mathcal M}}
\newcommand{\Nc}{{\mathcal N}}

\newcommand{\Pc}{{\mathcal P}}

\newcommand{\Rc}{{\mathcal R}}
\newcommand{\Sc}{{\mathcal S}}

\newcommand{\Uc}{{\mathcal U}}

\newcommand{\Vc}{{\mathcal V}}

\newcommand{\CN}[1][]{\ifthenelse{\isempty{#1}}{\mathcal{N}_{\mathbb{C}}}{\mathcal{N}_{\mathbb{C}}\left(#1\right)}}
\renewcommand{\P}[1][]{\ifthenelse{\isempty{#1}}{\mathbb{P}}{\mathbb{P}\left(#1\right)}}
\newcommand{\E}[1][]{\ifthenelse{\isempty{#1}}{\mathbb{E}}{\mathbb{E}\left(#1\right)}}
\renewcommand{\det}[1][]{\ifthenelse{\isempty{#1}}{\mathrm{det}}{\mathrm{det}\left(#1\right)}}
\newcommand{\trace}[1][]{\ifthenelse{\isempty{#1}}{\mathrm{tr}}{\mathrm{tr}\left(#1\right)}}
\newcommand{\rank}[1][]{\ifthenelse{\isempty{#1}}{\mathrm{rank}}{\mathrm{rank}\left(#1\right)}}
\newcommand{\diag}[1][]{\ifthenelse{\isempty{#1}}{\mathrm{diag}}{\mathrm{diag}\left(#1\right)}}
\newcommand{\blkdiag}[1][]{\ifthenelse{\isempty{#1}}{\mathrm{blkdiag}}{\mathrm{blkdiag}\left(#1\right)}}

\DeclarePairedDelimiter\abs{\lvert}{\rvert}
\DeclarePairedDelimiter\Abs{\lvert}{\rvert^2}

\renewcommand{\Re}{{\rm Re}}
\renewcommand{\Im}{{\rm Im}}

\newcommand{\herm}{{\sf H}}

\allowdisplaybreaks
\newcommand{\st}{{\rm s.t.}}
\DeclareMathAlphabet{\mathcal}{OMS}{cmsy}{m}{n}

\newcommand{\defeq}{\triangleq}

\newtheorem{remark}{Remark}
\newtheorem{theorem}{Theorem}
\newtheorem{lemma}{Lemma}

\usepackage{hyperref}
\hypersetup{
    bookmarks=true,         
    unicode=false,          
    pdftoolbar=true,        
    pdfmenubar=true,        
    pdffitwindow=false,     
    pdfstartview={FitH},    
    pdfnewwindow=true,      
    colorlinks=true,       
    linkcolor=red,          
    citecolor=green,        
    filecolor=magenta,      
    urlcolor=cyan           
}